\documentclass[12pt]{article}
\usepackage{amsmath}
\usepackage{graphicx}
\usepackage{enumerate}
\usepackage{natbib}
\usepackage{url} % not crucial - just used below for the URL 

\usepackage{hyperref}

\hypersetup{
  colorlinks,
  citecolor=black,
  linkcolor=black,
  urlcolor=black}

\usepackage{amsmath}
\usepackage{graphicx,amssymb,booktabs,multirow}
\usepackage{anysize}
\usepackage{bm}
\usepackage{natbib,amsmath}
\usepackage{graphicx}
\usepackage{fancyhdr}
\usepackage{rotating}

\makeatletter
\newcommand*{\indep}{%
	\mathbin{%
		\mathpalette{\@indep}{}%
	}%
}

\newcommand*{\nindep}{%
	\mathbin{% % The final symbol is a binary math operator
		\mathpalette{\@indep}{\not}% \mathpalette helps for the adaptation
	}%
}
\newcommand*{\@indep}[2]{%
	\sbox0{$#1\perp\m@th$}% box 0 contains \perp symbol
	\sbox2{$#1=$}% box 2 for the height of =
	\sbox4{$#1\vcenter{}$}% box 4 for the height of the math axis
	\rlap{\copy0}% first \perp
	\dimen@=\dimexpr\ht2-\ht4-.2pt\relax
	\kern\dimen@
	{#2}%
	\kern\dimen@
	\copy0 % second \perp
}
\makeatother

\usepackage{amsthm}

\makeatother
\newtheorem{theorem}{Theorem}
\newtheorem{lemma}{Lemma}

\newtheorem{proposition}{Proposition}
\newtheorem{assumption}{Assumption}

\usepackage{bbm}

\usepackage{caption}
\usepackage{subcaption}

\usepackage[all]{xy}
\usepackage{bm}
\usepackage{pifont}
\usepackage{stmaryrd}
\usepackage{pgf,tikz}
\usepackage{tikz}
\usetikzlibrary{shapes,arrows,positioning,calc}

\usepackage[flushleft]{threeparttable}
\usepackage{enumitem}
\usepackage{mathrsfs}

\def\pr{\textnormal{pr}}
\def\AT{{ss}}
\def\CO{{s\bar{s}}}
\def\NT{{\bar{s}\bar{s}}}
\def\DE{{\bar{s}s}}

\def\pr{\textnormal{pr}}

\makeatletter
\newcommand{\mathleft}{\@fleqntrue\@mathmargin0pt}
\newcommand{\mathcenter}{\@fleqnfalse}
\makeatother

\newtheorem{model}{Model}
\usepackage{algorithm}
\usepackage{algorithmic}
\usepackage{wrapfig}

\allowdisplaybreaks

\usepackage{accents}

\def\T{{ \mathrm{\scriptscriptstyle T} }}

\begin{document}

	\def\spacingset#1{\renewcommand{\baselinestretch}%
		{#1}\small\normalsize} \spacingset{1}

	%%%%%%%%%%%%%%%%%%%%%%%%%%%%%%%%%%%%%%%%%%%%%%%%%%%%%%%%%%%%%%%%%%%%%%%%%%%%%%
	
	%\if1\blind
	%{\title{}
		\title{\bf Identification and estimation of causal effects in the presence of confounded principal strata}
		\author{
Shanshan Luo\footnotemark[1], \;
Wei Li\footnotemark[2],\; %对应同一个脚注符号
 Wang Miao\footnotemark[3] \;%\Envelope
and Yangbo He\footnotemark[4] %\Envelope
}
\maketitle 
 \footnotetext[1]{School of Mathematical Sciences, Peking University. Email: \href{mailto:shan3\_luo@pku.edu.cn}{shan3\_luo@pku.edu.cn }} %对应脚注[1]
\footnotetext[2]{Center for Applied Statistics and School of Statistics,  Renmin University of China. Email: \href{mailto:weilistat@ruc.edu.cn}{weilistat@ruc.edu.cn}} 
 \footnotetext[3]{School of Mathematical Sciences, Peking University. Email: \href{mailto:mwfy@pku.edu.cn}{mwfy@pku.edu.cn}}  
 \footnotetext[4]{School of Mathematical Sciences, Peking University. Email: \href{mailto:heyb@math.pku.edu.cn}{heyb@math.pku.edu.cn}} 
 		 %		\author{Wei Li\thanks{Center for Applied Statistics and School of Statistics, Renmin University of China, Beijing 100872, China; email: \texttt{weilistat@ruc.edu.cn}.},~~Shanshan Luo\thanks{School of Mathematical Sciences, Peking University, Beijing 100871, China; email: \texttt{shan3\_luo@pku.edu.cn}.},~~and Wangli Xu\thanks{Corresponding author; Center for Applied Statistics and School of Statistics, Renmin University of China, Beijing 100872, China; email: \texttt{wlxu@ruc.edu.cn}.}}
		\date{}
% 		 \author{Wei Li
% 			%  	\thanks{
% 				%    The authors gratefully acknowledge \textit{please remember to list all relevant funding sources in the unblinded version}}
% 				\hspace{.2cm}
% 			\\
% 			    Center for Applied Statistics and School of Statistics, Renmin University of China\\
% 			    and \\
% 			    Author 2 \\
% 			    Department of ZZZ, University of WWW}
		\maketitle
		%} \fi
	
	%\if0\blind
	%{
		%  \bigskip
		%  \bigskip
		%  \bigskip
		%  \begin{center}
			%    {\LARGE\bf Title}
			%\end{center}
			%  \medskip
			%} \fi
		 
\begin{abstract}
 The principal stratification  has become a popular tool to address a broad class of causal inference questions, particularly in dealing with non-compliance and truncation-by-death problems.  The causal effects within principal strata which are determined by joint potential values of the intermediate variable, also known as the principal causal effects, are often of interest in these studies. Analyses of principal causal effects from observed data in the literature mostly rely on    ignorability of the treatment assignment, which requires practitioners to accurately measure as many as covariates so that all possible confounding sources are captured. However,  collecting all potential confounders in observational studies is often difficult and costly,   the ignorability assumption may thus be questionable. In this paper, by leveraging available negative controls that have been increasingly used to deal with uncontrolled  confounding, we consider identification and estimation of causal effects when the treatment and principal strata are confounded by unobserved variables. Specifically,  we show that the principal causal effects can be nonparametrically identified by invoking a pair of negative controls that are both required not to directly affect the outcome. We then relax this assumption and establish identification of principal causal effects under various semiparametric or parametric models. We also  propose an estimation method of principal causal effects.
Extensive simulation studies show good performance of the proposed approach and a real data application from the  National  Longitudinal  Survey of  Young  Men is used for illustration.
% by use of negative controls as the twofold role: (i) the proportions of principal strata can be recovered by integrating the confounding bridge with  negative controls (ii) principal causal effects can be  identified  by imposing some independence condition on negative controls. 
% When the identification assumption is invalid, we also consider an alternative set of  conditions for  identifying  principal causal effect in a semiparametric or parametric manner. We then propose  model parameterizations for estimation of  principal causal effects within the  identification assumptions. Our methods are illustrated via simulation studies and real data application. 
\end{abstract}
 
\vspace{1.5cm}\section{Introduction} 
%{In observational studies, researchers are often interested in understanding the underlying causal mechanism from the treatment to the outcome in the presence of intermediate variable. However,    even in randomized experiments, it is not feasible to  randomize the intermediate variable. Therefore, naive analysis by conditioning on the observed value of the intermediate variable is not causally interpretable. \citet{Frangakis:2002} propose to estimate causal effects within principal strata, which are defined by the joint potential values of the intermediate variable under both treatment and control. These problems with intermediate variables concern average causal effects within principal strata, which are also referred to as  principal causal effects ({principal causal effects}). Applications of {principal causal effects} have appeared in many  practical issues, such as,  non-compliance, truncation by death, missing data, mediation, and surrogate evaluation.}

% Investigators are often interested in analyzing the underlying causal mechanism 
Many scientific problems are concerned with evaluating the causal effect of a treatment on an outcome in the presence of an intermediate variable.
% Simple analysis by conditioning on the observed value of the intermediate  variable is not causally interpre, even after adjusting for a possibly large number of observed pretreatment covariates.   
{Direct comparisons} conditional on observed values of the intermediate variable are usually not causally interpretable.
\citet{Frangakis:2002}  {propose} the principal stratification   framework {and}  define principal {causal} effects that can effectively  compare {different treatment assignments} in such settings. The principal stratification is defined by
 joint potential values of the intermediate variable under each treatment level being compared, which is not affected by treatment assignment, and hence it can be viewed as a pretreatment covariate to
 classify individuals into subpopulations. Principal causal effects 
%  \wl{Is the acronym PCE  allowed by the styleguide?}
 that are defined as {potential outcome contrasts} within principal strata thus exhibit clear scientific {interpretations} in many  practical {studies} \citep{VanderWeele2011Principal}.  For instance, in non-compliance problems, the intermediate variable  is the actual treatment received, the principal stratification represents the compliance status, and the treatment assignment {plays the role of} an instrumental variable 
in identifying the complier average causal effect  \citep{Angrist:1996}. In truncation-by-death problems, the intermediate variable denotes survival status, and a meaningful parameter termed survivor average causal effect is defined as the effect among the subgroup who would survive under both treatment levels \citep{rubin2006causal,Zhang:2009,Ding:2011}.

Analysis of {principal causal effects} from observed data is challenging, because principal stratification is often viewed as an unobserved confounder between the intermediate and outcome variables.
%{\lss Principal stratification and treatment assignment can jointly  determine the value of the intermediate variable. Therefore the principal stratification variable can be regarded as the only unobserved  confounding between the intermediate and outcome variables. A}nalysis of {principal causal effects} from observed data is challenging {\lss since principal strata are inherently latent}. 
Most {works} in the literature rely on the {ignorability of treatment assignment}, which states that the distributions of potential values of intermediate and outcome variables do not vary across the treatment assignment given observed covariates. This assumption essentially 
requires that observed covariates account for all possible confounding factors between the treatment and post-treatment variables. Since the principal causal effects are defined on the latent principal {strata}, 
 one can only establish large sample bounds or conduct sensitivity analysis for  {principal causal effects} under the ignorability assumption 
\citep{Zhang:2003,Lee:2009,Long2013Sharpening}, but fails to
obtain identifiability results without additional assumptions. {Previous literature has used} an auxiliary variable that satisfies some conditional independence conditions to address the identification issues
\citep{Ding:2011,Jiang2016,DingPrincipal,Wang2017iv,Luo-multiarm-2021}.
% requires that a sufficiently rich set of covariates be measured such that given these pretreatment variables,}   the distributions of the potential values of intermediate variable and  outcome variable  do not vary across the observed treatment assignment. 
% Given all adequate confounders in observational studies, or even in randomized trials, {principal causal effects} cannot be identified without additional assumptions.  
% Large sample bounds or sensitivity analysis for  {principal causal effects} have been proposed  within   minimal conditions  \citep{Zhang:2003,Lee:2009, Long2013Sharpening,Luo-multiarm-2021}. %  In order to identify principal causal  effect, it is common to perform a sensitivity analysis by assuming a class of identification conditions indexed by a sensitivity parameter. 
% Under ignorability assumption, identification of {principal causal effects} is sometimes achievable with an  auxiliary covairiate  satisfying some conditional independence assumptions \citep{Ding:2011,Jiang2016,DingPrincipal,Wang2017iv}.}   
%All these results rely on the ignorability assumption.
{However, as can happen in observational studies, one may not sufficiently collect the pretreatment covariates.
%	%as often happens in practice, it is difficult to  detect and control for all possible confounding in observational studies. I
%it often happens  that one cannot sufficiently collect the pretreatment covariates in observational studies, 
 The existence of unmeasured variables may
% if the collected set of covariates is not rich enough, % enough to explain all common causes,  
% there may exist unmeasured confounding in the treatment assignments, 
% which 
render the ignorability assumption invalid and thus {the traditional causal estimates in principal stratification analysis can be biased.}}
%causal {\lss estimates} from traditional principal stratification analysis can be biased.}

%Because the principal stratum is unobservable and can sometimes be seen as a confounding, for example, it can reflect the different  compliance behavior of individuals in the context of non-compliance.}

As far as we know, there has not been much discussion on principal causal effects when the ignorability assumption fails. Several {authors} have considered the setting where the potential values of the intermediate variable are correlated with the treatment assignment even after conditioning on observed covariates. In other words, the treatment and the intermediate variable 
%latent principal strata 
are confounded {by}  unmeasured factors in this setting.
% {\lss For example,  \citet{schwartz2012sensitivity} consider the sensitivity analysis for unmeasured confounding in principal stratification settings when treatment has no direct effect on outcome.} 
\citet{schwartz2012sensitivity} present model-based approaches for assessing the sensitivity of complier average causal effect estimates in non-compliance problems when there exists unmeasured confounding in the treatment arms.
%{\lss For example, in non-compliance problems,  some literature consider the treatment assignment to be correlated with unobserved confounders, or equivalently, in the context of invalid instrumental variable. Specifically, 
%\citet{schwartz2012sensitivity}   consider the sensitivity analysis for the average treatment effects among compliers;}
\citet{KEDAGNI2021} discusses similar problems and
% allows the instrumental variable to be correlated with unobserved confounders and 
{provides} identifiability results
% of the complier average treatment effects 
by using a proxy for the  {confounded treatment assignment} %invalid instrument variable 
under some tail restrictions for the potential outcome distributions. \citet{deng2021identification}  study truncation-by-death problems and establish identification of the conditional average treatment effects for always-survivors given observed covariates by employing an auxiliary variable whose distribution is informative of principal strata. However, {because} the conditional distributions of principal strata given covariates are not identified, the survivor average causal effect is generally not identifiable in their setting.

% , which is fairly common in observational studies.  In such cases,  existing methods or techniques cannot even perform large sample bounds or conduct sensitivity analyses. From a practical point of view, it is necessary to develop approaches to handle this issue. 

% {In this paper, we consider  a setting in which the  the treatment and intermediate  processes are confounded  with unmeasured variables. Before our work, there are also some literature concerns principal causal effects, allowing some unmeasured confoundings to influence treatment assignment and intermediate variables.  
%  For example, in non-compliance problems, \citet{KEDAGNI2021} consider   identifying the local average treatment effects  for the compliers by using a proximal variable for the invalid instrument variable.  
%  In the context of truncation  by death, \citet{deng2021identification} study the identification of  conditional  average treatment  effects for the always-survivors by employing a auxiliary covariate  whose distribution is informative of  principal strata. 
 
{To overcome these limitations}, we establish identification of principal causal effects by leveraging {a pair of} negative control variables. In the absence of intermediate variables, many researchers have employed a negative control {exposure} and a negative control outcome to identify the average causal effects when unobserved confounders exist \citep{Miao2018Identifying,shi2020multiply,miao2020confounding,cui2020semiparametric}. However, 
 {the principal causal effects may be of more  interest in the presence of an intermediate variable.} 
% the average treatment effects may not always be of  interest  in the presence of intermediate variable.  %.
 For instance, in truncation-by-death problems,
 individuals may die before their outcome variables are measured, and hence the outcomes for dead individuals are not well defined. {Then  %{only scientifically meaningful parameter}
 %main parameter of interest 
 %in such studies is
 the {survivor average causal effect} is more scientifically meaningful in these studies \citep{rubin2006causal, tchetgen2014identification}}. While the identification and estimation of average causal effects within the negative control framework have been well studied in {the} literature, it remains uncultivated in studies where an intermediate variable exists and principal causal effects are of interest. 
 
 In this paper, we develop identification and estimation of principal causal effect {in the presence
of unmeasured confounders.}
 %a negative control treatment and negative control intermediate variable
  Specifically,
%  {There is a growing literature on the use of negative controls to mitigate confounding bias in   analysis of observational data.  In the absence of intermediate variables, many researchers have established nonparametric identification of causal effects by employing negative control variables as confounder proxies  \citep{Miao2018Identifying,shi2020multiply,miao2020confounding,cui2020semiparametric,tchetgen2020introduction}. One should not directly apply these existing approaches to 
% attenuate  confounding bias in the presence of intermediate variable, such as   truncation by death problem,   in which   the outcome variables  are not always well-defined. In this paper, we focus on the identification and estimation of causal effects confounded by unmeasured factor in the presence of intermediate variable.
we first introduce a confounding bridge function that links negative controls and the intermediate variable to identify proportions of the principal strata.  We then establish nonparametric identification of principal causal effects by assuming {that} the negative controls have no direct effect on {the} outcome.      We next relax this assumption and show alternative identifiability results based on semiparametric and parametric models. Finally, we provide an estimation {method} and discuss the asymptotic {properties}.  We evaluate the performance of the proposed estimator
{with}  simulation studies and a real data application.

 \section{Notation and assumptions}
 \label{sec: notation-assumption}
Assume that there are $n$ individuals {who are independent and identically sampled from a  superpopulation} of interest.
	Let $Z$ denote a binary treatment assignment with $1$ indicating treatment and $0$ for control. Let $Y$ denote an outcome of interest, and let $S$ denote {a}  binary intermediate variable.  Let $X$ denote a vector of covariates observed at baseline. 
	 We use the potential outcomes framework and make the stable unit treatment value assumption; that is, there is only one version of potential outcomes and there is no interference between units \citep{rubin1990comment}.
	%logarithmic weight difference between the end and beginning of the experiment
	Let $S_z$ and $Y_z$ denote the potential values of the intermediate variable and outcome that would be observed under treatment $Z=z$.  The observed values $S$ and $Y$ are deterministic functions of the treatment assignment and their respective potential values: $S=ZS_1+(1-Z)S_0$ and $Y=ZY_1+(1-Z)Y_0$.  	
	
	\citet{Frangakis:2002}  define the principal stratification as joint potential values of the intermediate variable under both the treatment and control.   We denote the basic principal stratum by $G$ and it can be expressed as $G=(S_{0},S_{1})$.	Since each of the potential values has two levels, there are four different principal strata in total.
	For simplicity, we refer to these principal strata, namely, $\{(0, 0),(0, 1),(1,1),(1, 0)\}$
	as never-takers ($\NT$), compliers ($\CO$), always-takers ($\AT$), and defiers ($\DE$),  respectively. % {\red Abbreviations are not allowed in Biometrika?}
%To address this issue, we consider sensitivity anslysis for this assumption later in Section  \ref{sec: sensitivit  
{The causal estimand of interest is the principal causal effect, i.e.,}
  $$~~~~~~~\Delta_{g}=E(Y_{1}-Y_{0}\mid G= g),\; g\in\{\NT,\CO,\DE,\AT\}. $$
  
  The {principal causal effect} conditional on a latent variable $G$ is  not identifiable without {additional} assumptions. Here we do not impose the  exclusion restriction assumption \citep{Angrist:1996} that requires {no}  individual causal effect on the outcome among the subpopulations $G=\AT$ and $G=\NT$, because in many settings with intermediate variables, such as {truncation-by-death} or surrogate problems \citep{gilbert2008evaluating}, the very scientific question of interest is to test whether the principal causal effect $\Delta_{\AT}$ or $\Delta_{\NT}$ is zero. {Under} this setup, the identification of $\Delta_{g}$ in the literature often relies on the following monotonicity assumption.

\begin{assumption}[Monotonicity]
\label{assumption: monotonicity}
  $S_1\geq S_0 $.
%   , i.e.,  $G\in\{ {\AT,\CO,\NT}\}$.
 \end{assumption}  
 %The monotonicity assumption may be plausible in some observational studies. {In section \ref{sec: sensitivity-analysis}, }Assumption \ref{assumption: monotonicity} implies that the treatment has a non-negative effect on the intermediate variable for all units, which rules out the stratum $G={\CO}$.  Besides, under assumption \ref{assumption: monotonicity},  we have the equivalence of $\{G={\AT}\}$ and $\{S_0=1\}$ as well as  $\{G={\NT}\}$ and $\{S_1=0\}$.
 
 Monotonicity rules out the existence of the defier group $G=\DE$. This assumption may be plausible in some observational studies. % if the treatment levels are reasonably labelled.%{二值的时候不用 reasonably labelled？就是S_1 \geq S_0或者 S_1 \leq S_0??} %\wl{For example, in our studies...}
 For example, in studies evaluating the effect of educational attainment on future earnings,   a subject living near a college is likely to receive  {{a} higher educational level}. The second commonly-used assumption is the treatment ignorability assumption:  $Z\indep   (S_{0},S_{1},Y_{0},Y_{1})\mid X$. This assumption entails that the baseline covariates $X$ control for all confounding factors between the treatment and post-treatment variables.
 However,  {the ignorability fails in the presence of unmeasured confounding. Let $U$ denote an} unobserved variable, which together with observed covariates $X$, captures all potential confounding sources between the {treatment  $Z$ and   variables $(S,Y)$}. We impose the following latent ignorability assumption. 
 
 \begin{assumption}[Latent ignorability]
\label{assumption: latent ignorability}
(i) $Z \indep (S_{0},S_{1})\mid (U, X )$; (ii)  $ Z\indep (Y_{0},Y_{1})\mid (G,X)$.
 \end{assumption} 
 
 The type of confounding {considered} in  Assumption~\ref{assumption: latent ignorability}{(i)} is termed $S$-confounding by \cite{schwartz2012sensitivity}.
The presence of the unmeasured variable $U$ in this assumption
%~\ref{assumption: latent ignorability}{(i)} has widely {broadened} the settings considered in most principal stratification literature, which, on the other hand, also 
brings about dramatic methodological changes and important technical challenges to principal stratification analysis. For example, when the traditional ignorability assumption {holds}, the  inequality $\pr(S=1\mid Z=1,X)< \pr(S=1\mid Z=0,X)$ can be used to falsify the monotonicity assumption. However, if  $U$ exists, 
it is no longer possible to empirically test monotonicity using this inequality.
%{\lss this inequality would no longer be used to empirically test for monotonicity.}
%it is {\lss no longer possible} to empirically {test} %falsify monotonicity based on this inequality.
In addition, {if we define principal score $\pi_{g}(X)$ as the proportion of the principal stratum given observed covariates \citep{DingPrincipal}, namely, $\pi_{g}(X)=\pr(G=g\mid X)$, the}
%if we denote the principal score which is defined as the proportion of principal stratum given   observed covariates \citep{jo2009use,DingPrincipal} by $\pi_g(X)$, i.e., $\pi_{g}(X)=\pr(G=g\mid X)$, the
presence of $U$ impedes identification of $\pi_g(X)$.  Assumption~\ref{assumption: latent ignorability}(ii) means that the confounding factors between the treatment and the outcome are fully characterized by the latent principal stratification $G$ and observed covariates $X$ {\citep{Wang2017iv}}.  %{\blue One may further relax this assumption by including an unobserved variable $U$, %but undoubtly, {as expected,} 
%the identification and estimation will become more challenged.   
Assumption~\ref{assumption: latent ignorability} {has} also {been} considered by 
\citet{KEDAGNI2021} and \citet{deng2021identification}.

We next discuss identification of $\Delta_{g}$ under Assumptions~\ref{assumption: monotonicity} and~\ref{assumption: latent ignorability}.
For simplicity, we define $\mu_{z,g}=E(Y_z \mid G=g)$, and hence $\Delta_g=\mu_{1,g}-\mu_{0,g}$. It {suffices} to identify $\mu_{z,g}$ for 
the identification of $\Delta_g$. Let  $\mu_{z,g}( X)=E(Y \mid Z=z, G=g, X)$. Then
% \begin{equation*}
%  \pi_{g}(X)=\pr(G=g\mid X), \; \;   \mu_{z,g}( X)=E(Y \mid Z=z, G=g, X) , 
% \end{equation*}
% where $\pi_g(X)$ is  known as the   principal score which denotes the proportion of principal strata given all   observed covariates \citep{jo2009use,DingPrincipal}, and $ \mu_{z,g}( X)$ is the   outcome mean conditional on the subpopulation $(Z=z,G=g,X)$. Let $\mu_{z,g}=E(Y_z \mid G=g)$, 
 under Assumption \ref{assumption: latent ignorability}(ii), we have 
\begin{equation}
\label{eq: the-expression-of-PCE}
\mu_{z,g}={E\{\mu_{z,g}( X)\pi_{g}(X)\}}/{E\{\pi_{g}(X)\}}.
\end{equation} 
It can be seen  that  the identification of $\mu_{z,g}$ depends on that of $\pi_{g}(X)$ and $\mu_{z,g}( X)$. Under Assumptions  \ref{assumption: monotonicity} and \ref{assumption: latent ignorability}(i), we have that 
\begin{align*} 
% \begin{array}{l}
\pi_{{\AT}}(X)&=E\left\{p_{0}(X, U) \mid X\right\},\quad\pi_{{\NT}}(X)=1-E\left\{p_{1}(X, U) \mid X\right\},\\
 &~~~~\pi_{{\CO}}(X)=E\left\{p_{1}(X, U)-p_{0}(X, U) \mid X\right\},
% \end{array}
\end{align*}
 where $p_z(X,U)=\pr(S=1\mid Z=z,X,U)$.
 {Because} $U$ is unobserved, the principal scores in the above equations cannot be identified without {additional} {assumptions}. As for the conditional outcome means $\{\mu_{z,g}(X):z=0,1;g=\AT,\NT,\CO\}$, only $\mu_{0,\AT}(X)$ and $\mu_{1,\NT}( X)$ can be identified under Assumptions~\ref{assumption: monotonicity} and~\ref{assumption: latent ignorability}(ii)  by $
\mu_{0,\AT}( X)=E(Y \mid Z=0, S=1, X)$ and $\mu_{1,\NT}( X)=E(Y \mid Z=1, S=0, X)$. However, the identifiability of other conditional outcome means is not guaranteed,  {because} the observed data $(Z=1,S=1,X)$ and $(Z=0,S=0,X)$ are {mixtures of} two principal strata:
 \begin{equation}
\label{eq: the-mixture-expression}
\begin{aligned} 
E(Y \mid Z=1, S=1, X)={{\textstyle\sum}_{g={\AT},{\CO}}\eta_{g}(1,X) \mu_{1,g}(X)},\\ 
E(Y \mid Z=0, S=0, X)={\textstyle\sum}_{g={\CO},{\NT}}\eta_{g}(0,X)  \mu_{0,g}(X),
\end{aligned}
\end{equation} 
where $\eta_{g}(1,X)=\omega_g(1,X) /\{\omega_{\AT}(1,X)+\omega_{\CO}(1,X)\}$, $\eta_{g}(0,X)=\omega_g(0,X) /\{\omega_{\NT}(0,X)+\omega_{\CO}(0,X)\}$, and  $   \omega_{g}(z,X)=\pr(G=g\mid Z=z,X) $. In later sections, the conditional probabilities of principal strata given only a subset $V$ of covarites $X$ may be of interest, and we simply denote them by replacing $X$ with $V$ in the original notations. For example, $\pi_g(V)=\pr(G=g\mid V)$. Other notations, such as $\omega_{g}(z,V)$ and $\eta_g(z,V)$, can be similarly interpreted. 
% {\blue  $   \eta_{g}(z,V)=\omega_g(z,V) /\pr(S=z\mid Z=z,V) $ and  $   \omega_{g}(z,X)=\pr(G=g\mid Z=z,V) $  for a generic variable $V$ with $z=0,1$}. 
%{ $   \eta_{g}(z,X)=\omega_g(z,X) /\pr(S=z\mid Z=z,X) $ and  $   \omega_{g}(z,X)=\pr(G=g\mid Z=z,X) $  for $z=0,1$}
{Due to} the {presence} of unobserved confounders $U$, the weights $   \eta_{g}(z,X)$  {in \eqref{eq: the-mixture-expression}} are no longer identifiable, which complicates the identification  and differs from most of existing results in the literature. 
{In such a case,} %Under the current circumstance, 
the large sample bounds or sensitivity analysis for these conditional outcome means cannot be easily obtained  without further assumptions
% \citep{Zhang:2003,IMAI2008144,Long2013Sharpening}, 
and  it would be even more difficult to obtain their identifiability results. In the following section, we discuss how to establish the identifiability of principal causal effects based on auxiliary variables.

 \section{Identification} 
 \subsection{Nonparametric  identification   using a pair of  negative controls}
 %In this section, we consider a nonparametric identification strategy of {principal causal effects}.
 {%From   previous discussion, we know that identification of  {principal causal effects} depends on the  quantities  $\omega_g(Z,X)$, $\pi_g(X)$ and $\mu_{z,g}(X)$. 
 In this section, we establish a nonparametric identification result for {principal causal effects}  through a pair of negative control variables when the ignorability assumption fails.} %Specifically, our goal is to mitigate confounding bias induced by the unobserved variable $U$ so that the  weights of principal strata are identifiable and further obtain the identifiability of principal causal effects.}
    Motivated from the proximal causal {inference} framework for identifying average treatment effects \citep{Miao2018Identifying,shi2020multiply,miao2020confounding,cui2020semiparametric}, we  assume that the covariates  $X $ can be decomposed into $(A,W,C^{\T})$ such that
  $A$ serves as a negative control {exposure}, $W$ serves as a negative control {intermediate} variable and  $C$
  accounts for the remaining observed confounders. For convenience, we may use the notation		$X$  and  $(A,W,C^{\T})$ interchangeably below.  
  %We condition on covariates $C$ implicitly and omit them below. 
%   Formally, we   introduce the following negative control framework.
%  \citep{Miao2018Identifying,shi2020multiply,miao2020confounding,cui2020semiparametric}.   
   \begin{assumption}[Negative control]\label{assumption: negative-control-variables} 
 $(Z,A)\indep (S_0,S_1,W)\mid (C,U)$ 
 \end{assumption} 
    \begin{assumption}[Confounding bridge]\label{assumption:  bridge-fun}
There exists a function $h(z, W,C) $ such that $ 
\pr(S=1\mid Z=z,C,U)=E\{ h(z, W,C) \mid  C,U \}$
 almost surely for {all} $z$. 
 %$ A \indep S \mid (U, C),W \indep(A, C)\mid U.$
 \end{assumption}

\begin{figure}[t]   
			\centering \includegraphics[width=0.45\linewidth]{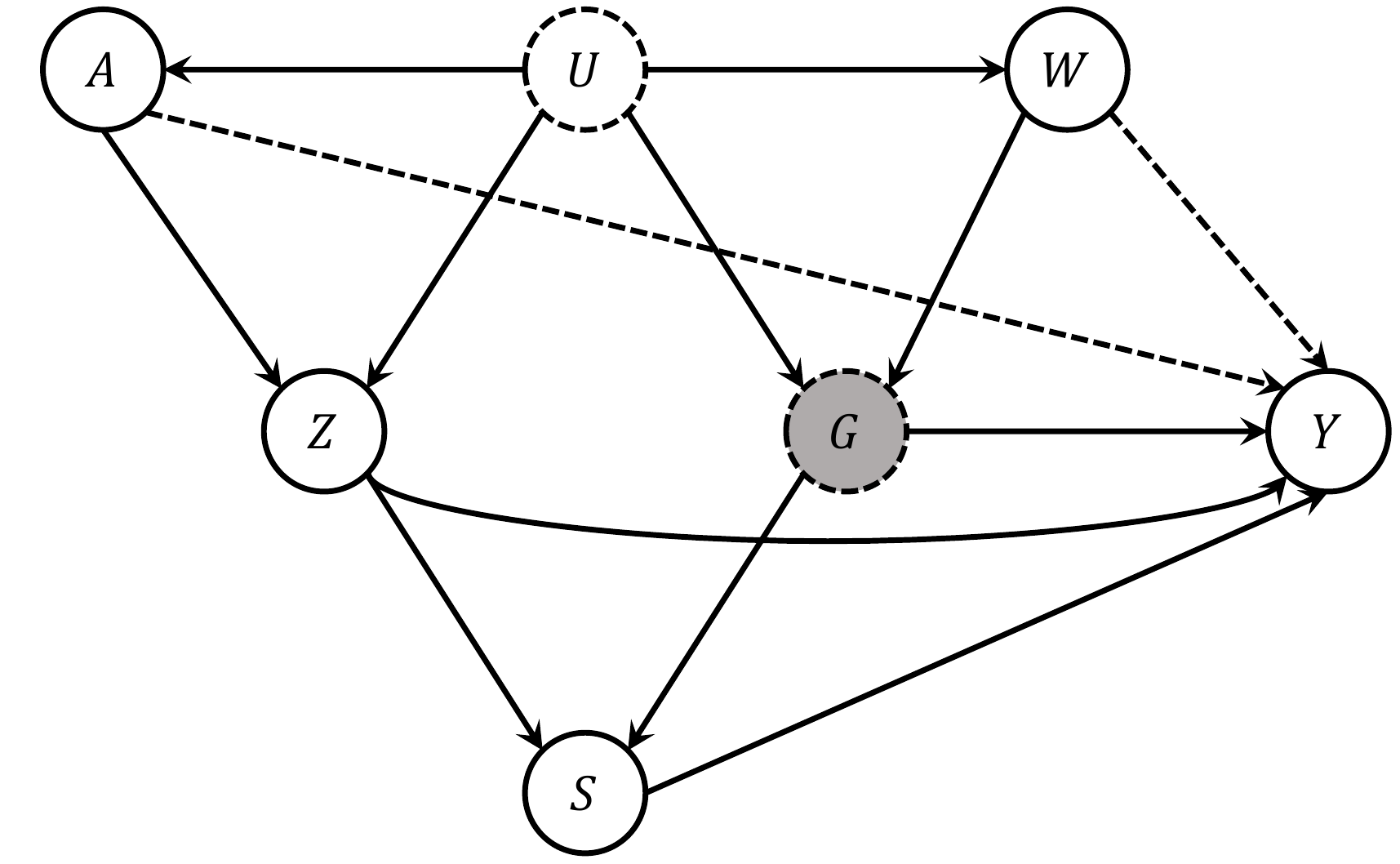} 
	%		\caption{$(Z,A)\indep (S_z,W)\mid (C,U)$}  
    			\caption{A causal diagram  illustrating treatment and intermediate variable confounding proxies when ignorability assumption fails.    Dashed {arrows} indicate edges can exist when {semiparametric or parametric models are considered.}  Observed covariates are omitted for simplicity. }
% 			\wl{Is it necessary to draw Figure 1(a)?}}
		\label{fig: NG-graph}
	\end{figure}

Assumption \ref{assumption: negative-control-variables} implies that the variables  $(C,U)$ are sufficient to  account for  the {confounding} between $(Z,A)$ and $(S_0,S_1,W)$. 
The negative control exposure $A$ {does} not directly affect either the intermediate variable $S$ or the negative control intermediate variable $W$. 
% Monotonicity and 
Assumption \ref{assumption: negative-control-variables} imposes no restrictions on   $Z$-$A$ association or $G$-$W$ association, and    allows the two negative controls  $A$ and $W$ to be confounded by the unmeasured variable $U$.
 %Combined with monotonicity assumption \ref{assumption: monotonicity},  Assumption \ref{assumption: negative-control-variables}  indicates that 
  % $(C,U)$  includes all unmeasured
%common causes of $Z,A, G$ and $W$ except for that of the $Z$-$A$ association and $G$-$W$ association. 
See Fig. \ref{fig: NG-graph}  for a graphic illustration. 
{The confounding bridge function in Assumption \ref{assumption:  bridge-fun} establishes the connection between the negative control  $W$ and the intermediate variable $S$.} Assumption \ref{assumption:  bridge-fun}    defines  an   inverse problem known as {the} Fredholm integral equation of the first kind. 
  The technical conditions for the existence of a solution are provided in \citet{Carrasco2007Fredholm}.
%   , \citet{Miao2018Identifying} and \citet{cui2020semiparametric}. 
Since the principal stratum $G$ is a latent variable, Assumptions~\ref{assumption: negative-control-variables} and \ref{assumption:  bridge-fun}, {which are used to control for unobserved confounding between treatment and intermediate variables,} are not sufficient to nonparametrically identify principal causal effects.
  We thus impose  the following conditional independence condition between the negative controls and  potential outcomes given the latent variable $G$ and observed covariates $C$.
%   , as given in Assumption \ref{assump: negative-control-indep}.
 \begin{assumption}
  \label{assump: negative-control-indep}
  $(Z,A,W)\indep (Y_0,Y_1) \mid  (G   ,C)$.
  \end{assumption}
  
%  	Assumption \ref{assump: negative-control-indep} precludes the direct effect of the pair of negative controls $A$ and $W$ on outcome $Y$.
 {Under} Assumption \ref{assump: negative-control-indep}, {we can} view the observed  variables $A$ and $W$ as   proxies of $G$, the role of which resembles the usual instrumental variables that preclude direct effects on the outcome $Y$ \citep{Angrist:1996}; see Fig. \ref{fig: NG-graph} for an illustration. Similar assumptions have been widely used in principal stratification literature \citep{Ding:2011,Jiang2016,Wang2017iv,Luo-multiarm-2021}.
%  	see Fig \ref{fig: NG-graph} for 
%  	, which is similar to the exclusion restriction in the instrumental variable analysis \citep{Angrist:1996}. Fig.  \ref{fig: NG-graph}(b) gives a causal diagram illustrating Assumption \ref{assump: negative-control-indep}. 
%  	The main difficulty with the identifiability problem is that the principal strata are inherently latent. Assumption   \ref{assump: negative-control-indep} makes it possible to identify {principal causal effects},   %since the observed proxies $A$ and $W$ can be used to recover the information of  principal strata.  
%  	since the observed proximal variables $A$ and $W$ can be viewed as substitutional variables for the confounded treatment assignment $Z$ and latent principal stratum $G$, respectively.  
 	In the next subsection,	we shall  consider to relax this assumption based on semiparametric or parameric models.
 \begin{theorem}
 \label{thm: PCE-identification-AC}
 Suppose that Assumptions \ref{assumption: monotonicity}, \ref{assumption: latent ignorability}(i),  \ref{assumption: negative-control-variables} and \ref{assumption:  bridge-fun} hold. Then the {conditional probabilities} of principal strata 
 are  identified by 
\begin{equation}
\label{eq: principal-score-A-C}
 \begin{array}{lcl}\omega_{{\AT}}(Z,A,C)= E\{h(0, W,C)\mid Z,A,C\},\\
\omega_{{\NT}}(Z,A,C)= 1-E\{h(1, W,C)\mid  Z,A,C\},\\
\omega_{{\CO}}(Z,A,C)=1-\omega_{{\AT}}(Z,A,C)-\omega_{{\NT}}(Z,A,C).\\
 % \int_{-\infty}^{+\infty}h(z, W,C)f(w\mid a,c)dw
% {\blue \operatorname{pr}(S_t=1\mid Z=z ,A,C)=E\{h(t,W,C)\mid Z=z ,A,C\} .}
 \end{array}
\end{equation}   
{Under   additional Assumptions \ref{assumption: latent ignorability}(ii) and \ref{assump: negative-control-indep}, the {principal causal effects} are  identifiable if  for any $C=c$, the functions in the following two vectors	    \begin{equation}
	        \label{eq: indep-1}
 \{\eta_{{\CO}}(0,A,c),\eta_{{\NT}}(0,A,c)\}^\T\quad\text{and}\quad \{\eta_{{\CO}}(1,A,c),\eta_{{\AT}}(1,A,c)\}^\T
	    \end{equation} are respectively linearly independent. }

 \end{theorem} 
  
  The identifiability result \eqref{eq: principal-score-A-C} in
 Theorem \ref{thm: PCE-identification-AC} links the confounding bridge function with the conditional probabilities of principal strata given  observed variables $(Z,A,C)$. {With an additional completeness condition, the bridge function in Assumption \ref{assumption:  bridge-fun} can be equivalently characterized by a solution to an equation based on observed variables (see Lemma \ref{lem: complete} in the supplementary materials).}
 Note that Assumption \ref{assumption:  bridge-fun} only requires the existence of solutions to the integral equation. Theorem \ref{thm: PCE-identification-AC} implies that even if $h(z,W,C)$ is not unique, all solutions to Assumption \ref{assumption:  bridge-fun} must result in an identical value of each conditional  proportion of the principal stratification.  
 
%  , equation \eqref{eq: principal-score-A-C} shows that the proportions of principal strata are  identifiable conditioning on the observed variables $(Z,A,C)$.  
% The   result links the confounding bridge function and proportions of principal stratification, which is motivated from \citet{Miao2018Identifying,miao2020confounding}.  In the absence of intermediate variable, \citet{miao2020confounding} also propose a confounding bridge approach for negative control inference on causal effects. 
%   Note that  equation \eqref{eq: principal-score-A-C} does not require the uniqueness of a solution to the integral equation in Assumption \ref{assumption:  bridge-fun}, all solutions lead to    the uniqueness of   $\omega_{g}(Z,A,C)$. 
    {In the absence of unmeasured confounding, \citet{Ding:2011} and \citet{Wang2017iv}  use only one {proxy} variable whose distribution is informative of principal stratum $G$ to establish nonparametric identification.  When the   principal strata are confounded by the unmeasured variable $U$,  
  Theorem \ref{thm: PCE-identification-AC} shows that {principal causal effects}  can also be identified with two {proxy}  variables.}  
The conditions in  \eqref{eq: indep-1}  are similar to the  relevance assumption in instrumental variable analyses \citep{Angrist:1996}, which requires the {association} between the negative control exposure $A$ and principal stratum $G$. Because the weights $\omega_g(z,A,c)$'s   are identified according to  \eqref{eq: principal-score-A-C},  the linear independence conditions among functions in each vector of \eqref{eq: indep-1} are in {principle} testable based on observed data.

% by the sample analogues of equation \eqref{eq: indep-1}.   
%	 Given the conditions in Theorem \ref{thm: PCE-identification-AC},  $ \mu_{z,g}$ can be identified from the following  equation (see the supplementary material for proof):
%\begin{equation}
%\label{eq: the-expression-of-PCE-AC}
%E(Y_z \mid G=g)={E\{\mu_{z,g}(A,C)\pi_{g}(A,C)\}}/{E\{\pi_{g}(A,C)\}},
%\end{equation}  
%where the principal score  $\pi_{g}(A,C)$ for all $g$ can be identified through Lemma \ref{lem: negative-control} and $\mu_{z,g}(A,C)$ can be nonparametrically or semiparametrically identified.   
 \subsection{Identification under semiparametric or parametric models} 
%Given Assumptions \ref{assumption: monotonicity} and \ref{assumption: negative-control-variables}, we can only nonparametrically identify  $\omega(Z,A,C)=\pr(G\mid Z,A,C)$ in Lemma \ref{lem: negative-control}. In order to  nonparametriclly identify the principal causal effect, we further require two negative control variables cannot directly affect $Y$ as shown in Assumption \ref{assump: negative-control-indep}.
  
In this section, we {relax} Assumption \ref{assump: negative-control-indep} to some extent and  discuss the identifiability of {principal causal effects} under semiparametric or parametric models.

	   	\begin{assumption}				\label{assumption: indep-of-W} 
		  $(Z,W)\indep (Y_0,Y_1) \mid(  G   ,C,A)$,  
	\end{assumption}  
	
	Assumption \ref{assumption: indep-of-W} is notably weaker than Assumption \ref{assump: negative-control-indep} by requiring only one substitutional variable $W$ for the principal stratum $G$, which allows negative control exposure $A$ to directly affect the outcome. This is in parallel to the usual assumption for the identifiability of principal causal effects when unmeasured confounding is absent \citep{Ding:2011,Jiang2016,Wang2017iv,Luo-multiarm-2021}. 
% 	The role of  $W$ in this assumption   is also  conceptually similar to the instrumental variable  \citep{Angrist:1996}. 
% Assumption \ref{assumption: indep-of-W}   precludes the direct effect of proximal variable $W$   on outcome $Y$, but allow  direct effect from $A$ to $Y$.   Assumption	    \ref{assumption: indep-of-W} is a widely used    identification assumption for {principal causal effects}   in the absence of unmeasured confounding $U$, where the auxiliary covariate $W$  can be viewed as a  substitutional variable for the latent  principal strata  \citep{Ding:2011,Jiang2016,Wang2017iv,Luo-multiarm-2021}.   
%{\lss Note that Assumption \ref{assumption: indep-of-W} has relaxed Assumption~\ref{assump: negative-control-indep} by allowing negative control treatment $A$ to directly affect the outcome,}  
We consider   a semiparametric linear model for $\mu_{z,g}(X)$   to {facilitate identification of} {principal causal effects} under Assumption \ref{assumption: indep-of-W}.
\begin{theorem}
\label{thm: linear-model-AC}
Suppose that Assumptions \ref{assumption: monotonicity}--\ref{assumption:  bridge-fun} and \ref{assumption: indep-of-W} hold.  We further assume $\mu_{z,g}(X)$ follows a  linear model:
   \begin{equation}
   \label{eq: linear-model}
\mu_{z,g}(X)=\theta_{z,g,0 }+ {\theta}_{c  }C +\theta_{a} A.
\end{equation}
Then the {principal causal effects} are  identified if   the functions in these two vectors
$$
   \left\{\eta_{{\CO}}(0,A,C),\eta_{{\NT}}(0,A,C),A,C\right\}^{\T} \;\text{and}\; \left\{\eta_{{\CO}}(1,A,C),\eta_{{\AT}}(1,A,C),A,C\right\}^{\T}  
$$ 
are respectively linearly {independent}. 
\end{theorem}
%Note that whether Theorem \ref{thm: PCE-identification-AC} or \ref{thm: linear-model-AC}, we need to assume that     Assumption \ref{assumption: indep-of-A} holds, which means that we need the negative control intermediate variable $W$ have no direct effect to outcome $Y$.  

% From Theorems \ref{thm: PCE-identification-AC}-\ref{thm: linear-model-AC}, we have established the identifiability of {principal causal effects}  via  Assumption \ref{assump: negative-control-indep}  or under a linear model setting. %Note that for both theorems, we need to assume   Assumption \ref{assumption: indep-of-W}, which requires that the proxy variable $W$   have no direct effect on outcome.  

 One may further relax Assumption \ref{assumption: indep-of-W} by
% Note that we assume Assumption \ref{assumption: indep-of-W} in both theorems,  which requires that the proximal variable $W$   have no direct effect on outcome. 
%The reason is that we can only identify $\omega_g(Z,A,C)$ nonparametrically in the first stage  but not $\omega_g(Z,X)$ as shown in Lemma \ref{lem: negative-control}.  
%In order to further  relax Assumption \ref{assumption: indep-of-A}, we consider identify the proportions of the principal strata  $\omega_g(Z,X)$ among all observed covariates $X$ as well as $Z$ via some parametric  models. 
% In contrast to \eqref{eq: linear-model}, one may also 
considering the following  model,
 \begin{equation}
   \label{eq: linear-model-2}
\mu_{z,g}(X)=\theta_{z,g,0}+{\theta}_{c  } {C}+\theta_{a} A+\theta_{w} W,
\end{equation} 
which allows the  outcome to be  affected by all   observed covariates $X$, including the negative control intermediate variable  $W$. {The above semiparametric linear model \eqref{eq: linear-model-2}  has also been  considered in \citet{Ding:2011}
% , \citet{Wang2017iv}
and \citet{Luo-multiarm-2021}.} As shown in the supplementary material, 
% given the conditions in Theorem \ref{thm: linear-model-AC}, 
the parameters in \eqref{eq: linear-model-2} are  identifiable under some regularity conditions. This means that we can identify the conditional outcome mean  $\mu_{z,g}(X)$. However, since the proportions of principal strata $\pi_g(X)$ are not identifiable under  the assumptions in  Theorem \ref{thm: linear-model-AC}, the parameter $\mu_{z,g}$ expressed in \eqref{eq: the-expression-of-PCE} cannot be
identified unless additional conditions exist. Below we consider parametric models that would make it possible to identify the principal causal effects even {if} Assumption~\ref{assumption: indep-of-W} were violated.

% We would focus on linear model \eqref{eq: linear-model-2} again later and show that {principal causal effects} are also identifiable  with additional conditions. To further relax Assumption \ref{assumption: indep-of-W}, we    consider identifying the proportions of  principal strata   conditioning on treatment $Z$ and   covariates $X$  via some parametric models.    
\begin{proposition}
\label{prop: idetification-of-principal-stratum} 
Suppose that Assumptions \ref{assumption: monotonicity}, \ref{assumption: latent ignorability}(i),  \ref{assumption: negative-control-variables} and \ref{assumption:  bridge-fun} hold.  The principal stratum $G$ follows an ordered probit model, namely,   
\begin{equation}
    \label{eq: ordered-probit-model}
    G=\left\{\begin{array}{rcl}{\NT},&&~~\mathrm{if}\;G^{\ast}+\varepsilon\leq 0,\\{\CO},&&~~\mathrm{if}\;0<G^{\ast}+\varepsilon\leq \exp({\psi_1}),\\{\AT},&&~~\mathrm{if}\;\exp({\psi_1})<G^{\ast}+\varepsilon,\end{array}\right.
\end{equation}
where $G^{\ast}=\psi_{0}+\psi_{z}Z+\psi_{w}W+\psi_{a}A+\psi_{c}C$ and $\varepsilon\sim N(0,1)$.  
We further assume that $W\mid Z,A,C\sim N\{{m}(Z,A,C),\sigma^2_w\} $ 
and the functions $\{1,Z,A,C,m(Z,A,C)\}{^\T}$  are linearly independent.  
Then the proportions of principal strata  $\omega_g(Z,X)$  are identified for all $g$. 
\end{proposition} 
%We used an ordered probit model to the latent variable $G$ and all observed variable $X$ with $Z$.
%This relatively simple discrete model ensures that the monotonicity is established, but it also makes the marginal effects of different compliance types on the covariates the same.

{%  
As implied by the latent ignorability assumption, the association $Z\nindep G\mid X$ may occur in the presence of unobserved confounder $U$, so the coefficient $\psi_z$ in the model for $G$ after~\eqref{eq: ordered-probit-model} may not be {zero}. 
% This relatively simple ordinal model  is compatible with the monotonicity, and also makes the marginal effects of different compliance types on covariates the same.  
 The ordinal model in~\eqref{eq: ordered-probit-model} is compatible with monotonicity  assumption and can be rewritten in the following form
%  , and  makes the same marginal effects on covariates   for different compliance types.   
% Equation \eqref{eq: ordered-probit-model}  
% is equivalent to the following model parameterizations 
under this assumption:
\begin{equation}
\label{eq: equiv-weight}
    \begin{array}{rcl}\pr\left(S_0=1\mid Z,X;\psi\right)&=&\Phi\{\psi_0-\exp(\psi_1)+\psi_zZ+\psi_wW+\psi_aA+\psi_cC\},\\\pr\left(S_1=1\mid Z,X;\psi\right)&=&\Phi(\psi_0+\psi_zZ+\psi_wW+\psi_aA+\psi_cC),\end{array}
\end{equation} 
where  $\psi=(\psi_0,\psi_1,\psi_z,\psi_w,\psi_a,\psi_c)^\T$. In fact, we model the distribution of potential values of the intermediate variable using a generalized linear model, which is similar in spirit to the marginal and nested structural mean models proposed by  \citet{robins2000marginala}. 
Under such parametric models, Proposition~\ref{prop: idetification-of-principal-stratum} shows that we can identify the conditional proportions of the principal strata $\omega_g(Z, X)$ given all observed covariates.  This is a stronger result than that in Theorem~\ref{thm: PCE-identification-AC}, where only the conditional proportions of principal strata given covariates $(A, C)$ are identifiable. With this result, we can consider another weaker version of  Assumption \ref{assump: negative-control-indep},  which is in parallel to Assumption \ref{assumption: indep-of-W}.
  	\begin{assumption}				\label{assumption: indep-of-A}
      $(Z,A)\indep (Y_0,Y_1) \mid(  G   ,C,W)$.   
	\end{assumption}  
	
 This condition is  similar to the ``selection on types" assumption considered in \citet{KEDAGNI2021}, which entails that the negative control exposure $A$ has no direct effect on the outcome $Y$.  We next consider  identification of {principal causal effects} under the ordinal model \eqref{eq: ordered-probit-model} and other various  conditions.% Given Assumption  \ref{assumption: Y-ignorability}, our proposed framework is applicable  to a variety of topics on {principal causal effects} even in the presence of unobserved confounders. 
%	In this paper, we restrict our attention to identifyidentifying the {principal causal effects}, more specifically,  we first	impose some independent conditions for negative control variables.
	\begin{theorem} 
	\label{thm: PCE-identification-X}
Under Assumptions \ref{assumption: monotonicity}--\ref{assumption:  bridge-fun} and the model parameterization in Proposition \ref{prop: idetification-of-principal-stratum},  the following statements hold: 
	\begin{itemize}
	 \item[(i)] with additional Assumption \ref{assump: negative-control-indep},  the {principal causal effects}    are   identified if  for any $C=c$,     the functions in the  vectors
	 $
 \{\eta_{{\CO}}(0,A,W,c),\eta_{{\NT}}(0,A,W,c)\}^\T$ and  $\{\eta_{{\CO}}(1,A,W,c),\\\eta_{{\AT}}(1,A,W,c)\}^\T $ 
	    are respectively linearly independent.
  
	    \item[(ii)] with additional Assumption \ref{assumption: indep-of-W}, the {principal causal effects}  are   identified if   for any $(A,C)=(a,c)$,  the functions in the vectors $
     \{\eta_{{\CO}}(0,a,W,c),\eta_{{\NT}}(0,a,W,c)\}^\T$ and $ \{\eta_{{\CO}}(1,a,W,c),\eta_{{\AT}}(1,a,W,c)\}^\T$
	    are respectively linearly independent.
	    
\item[(iii)] with additional Assumption \ref{assumption: indep-of-A},   the {principal causal effects}    are   identified if  for any $(W,C)=(w,c)$, the functions in the vectors  $
     \{\eta_{{\CO}}(0,A,w,c),\eta_{{\NT}}(0,A,w,c)\}^\T$ and $ \{\eta_{{\CO}}(1,A,w,c),\eta_{{\AT}}(1,A,w,c)\}^\T$
	    are respectively linearly independent.
 %the conditions in \eqref{eq: linear-nindep-1} are linearly independent   vectors of  $(A,C)$.
   \item[(iv)] %Suppose  $\mu_{z,g}(X)$ follows a linear model:
  with the additional    model  \eqref{eq: linear-model-2}, the {principal causal effects}  are   identified
 if the functions in the vectors   $
   \left\{\eta_{{\CO}}(0,X),\eta_{{\NT}}(0,X),X\right\}^{\T}$ and $\left\{\eta_{{\CO}}(1,X),\eta_{{\AT}}(1,X),X\right\}^{\T}  $
are respectively linearly independent.
	\end{itemize} 
	\end{theorem} 

In contrast to Theorem~\ref{thm: PCE-identification-AC},	
Theorem \ref{thm: PCE-identification-X} shows that under the parametric models in
 Proposition \ref{prop: idetification-of-principal-stratum},  the {principal causal effects} are always identifiable as long as certain linear independence conditions are satisfied, albeit
 Assumption~\ref{assump: negative-control-indep} may partially or completely fails.
% In   Theorem \ref{thm: PCE-identification-X}(d), linear model  \eqref{eq: linear-model-2} includes all observed covariates as predictors,  and  Assumption \ref{assump: negative-control-indep} requires $\theta_a$ and $\theta_w$ to be null  simultaneously in this linear  model.
Besides, since the functions $\{\eta_g(z,X);z=0,1;g=\AT,\NT,\CO\}$ are identifiable based on Proposition \ref{prop: idetification-of-principal-stratum},
those linear  independence conditions in Theorem \ref{thm: PCE-identification-X} {are} testable from observed data.
% , since the proportions $\omega_g(Z,X)$ and $\eta_g(Z,X)$ can be identified by Proposition \ref{prop: idetification-of-principal-stratum}. %{\red adding more comments, compared to the condition \cite{Luo-multiarm-2021}...} 
	  } 

 	\section{Estimation}
 	\label{sec: estimation}
  While the nonparametric identification results provide useful insight, nonparametric estimation, however, is often not practical especially when the number of covariates is large due to the curse of dimensionality.  We consider  {parametric working models} for estimation of {principal causal effects} 
%   under various conditions 
  in this section. 
%   We first introduce the parametric models involved in the   treatment and intermediate variable modeling process.
%We first consider following parametric models. %Although  the integral functions  can be addressed nonparametrically, it may not be feasible in practice.   {\red xxx}: 

%Our estimation procedure proceed as follows:  we first obtain the estimators of $\omega_g(Z,A,C)$ and $\pi_g(A,C)$  with  a parametric model of  bridge function in Lemma \ref{lem: negative-control}; next, provided the conditions in Proposition \ref{prop: idetification-of-principal-stratum}, we can further obtain the estimates for   $\omega_g(Z,X)$ and $\pi_g(X)$  by using the estimate of  $\omega_g(Z,A,C)$ from  previous step; finally, given conditions in Theorem \ref{thm: PCE-identification-AC}-\ref{thm: PCE-identification-X}, we estimate {principal causal effects}  based on the weights obtained in the first two steps.

%Our estimation procedure proceeds as follows: we first obtain the estimators of $\omega_g(Z,A,C)$ and $\pi_g(A,C)$  using the  parametric model of  bridge function in Lemma \ref{lem: negative-control}; next, under the conditions of Proposition 1,   the  estimators  $\omega_g(Z,A,C)$  in the previous step can be used to further obtain the estimators  of $\omega_g(Z,X)$ and $\pi_g(X)$; Finally, given the conditions in Theorem \ref{thm: PCE-identification-AC}-\ref{thm: PCE-identification-X}, we can   estimate the  {principal causal effects}  based on the weights obtained in the first two steps.
\begin{model}[Bridge function]
\label{model: bridge-function}
The bridge function $h(Z,W,C ;\alpha)$ is known up to a finite-dimensional  parameter $\alpha$.  
%  {We further restrict  bridge function such that $h(1,W,C ;\alpha)\geqh(0,W,C ;\alpha)$, which is required by the monotonicity assumption  (see the supplementary material for details). }  
\end{model} 
\begin{model}[Treatment and negative control intermediate variable] 
\label{model: distribution-of-models}	The treatment model $\pr(Z\mid   A,C;\beta)$ is known up to a finite-dimensional parameter $\beta$, and the negative control intermediate variable model $f(W\mid Z,  A,C;\gamma)$ is known up to a finite-dimensional parameter $\gamma$. %\wl{the parameters $\beta$ and $\gamma$ have been exchanged.}
% We consider the following distributional assumptions:
% \begin{itemize}
% 	\item[(a)]
% The probability $f(W\mid Z,  A,C)$ is known up to a finite-dimensional  model with  an unknown  parameter $\beta$.
% \item[(b)] The probability $\pr(Z\mid   A,C)$ is known up to finite-dimensional  model with  an unknown  parameter $\gamma$.
% \end{itemize}    
	\end{model} 
% 	We also posit some parametric models  for the conditional expectation $\mu_{z,g}(X)$.
	 \begin{model}[Outcome]
	 \label{model: outcome-pce}
The conditional outcome mean function $\mu_{z,g}(X;\theta_{z,g})$ is known up to a finite-dimensional 
% model with an unknown 
parameter $\theta_{z,g}$. %\wl{add some comment below; use one example to illustrate this model parameterization; no need to list all settings}
% \label{model: outcome regression} 
% \begin{itemize} 
% \item[(a)] Given the conditions in Theorem \ref{thm: PCE-identification-AC} or    \ref{thm: PCE-identification-X}(a), we consider
%  $ \mu_{z,g}(X) = \mu_{z,g}(C;\theta_{z,g}).$
% \item[(b)]  Given the conditions in Theorem \ref{thm: PCE-identification-X}(b), we consider
%   $\mu_{z,g}(X)= \mu_{z,g}(A,C;\theta_{z,g}).$ 
% \item[(c)] Given the conditions in Theorem \ref{thm: PCE-identification-X}(c), we consider
%  $ \mu_{z,g}(X) = \mu_{z,g}(W,C;\theta_{z,g}).$  
% \item[(d)]  Given the conditions in Theorem \ref{thm: linear-model-AC} or    \ref{thm: PCE-identification-X}(d),  we consider  model  \eqref{eq: linear-model} or \eqref{eq: linear-model-2}, respectively.
% \end{itemize}
\end{model}    
%  In practice, we can specify that the model  
 Note that   $\mu_{z,g}(X;\theta_{z,g})$ in Model  \ref{model: outcome-pce}  should be 
 compatible with the   requirements
 in Theorems \ref{thm: PCE-identification-AC}-\ref{thm: PCE-identification-X}. For example, under the conditions in Theorem \ref{thm: PCE-identification-AC} or    \ref{thm: PCE-identification-X}(i), we consider  a parametric form for $ \mu_{z,g}(X ;\theta_{z,g})$ that is only related to the covariate $C$, but should not be dependent on $A$ and $W$. Given the above parameterizations in Models~\ref{model: bridge-function}--\ref{model: outcome-pce}, we are  now ready to provide a three-step procedure for estimation of the principal causal effects.

In the first step, we   aim to estimate the conditional probabilities $\omega_g(Z,A,C)$  considered in Theorem \ref{thm: PCE-identification-AC}. 
The expression in \eqref{eq: principal-score-A-C} implies that for estimation of $ \omega_g(Z,A,C)$, we only need to  estimate the parameters $\alpha$ and $\gamma$ that are in the bridge function $h(Z,W,C;\alpha)$ and negative control intermediate variable model $f(W\mid Z,A,C;\gamma)$, respectively. 
% As shown in the supplementary material,  
Under Assumptions \ref{assumption: negative-control-variables},  \ref{assumption:  bridge-fun}, {and the completeness condition, we have the following equation} (see Lemma \ref{lem: complete}  in the    supplementary material for details):
$$\pr(S=1\mid Z,A,C)=E\{h(Z,W,C;\alpha)\mid Z,A,C\}.$$
We then obtain an estimator $\widehat{\alpha}$ by solving the following  estimating equations
\begin{equation}
    \label{eq: parameter-bridge-function}
    \mathbb{P}_n[\{S-h(Z,W,C;\alpha)\}B(Z,A,C)]=0,
\end{equation}
where $\mathbb{P}_n(\xi) = \sum_{i=1}^n\xi_i/n$ for some generic variable $\xi$, and $B(Z,A,C)$ is an arbitrary vector of functions  with dimension no smaller than that of $\alpha$. If the dimension of the user-specified function $B(Z,A,C)$ is larger than that of $\alpha$, we may adopt the generalized method of moments \citep{Hansen:1982} to estimate $\alpha$. We next obtain the estimators $\widehat\beta$ and $\widehat\gamma$ in Model \ref{model: distribution-of-models}  by maximum likelihood estimation.  With the parameter estimates $\widehat\alpha$ and $\widehat\gamma$, we can finally obtain the estimators $\widehat\omega_g(Z,A,C)$   based on~\eqref{eq: principal-score-A-C}. The calculation of the estimated probabilities involves  {integral equations} with respect to the distribution $f(W\mid Z,A,C; \widehat{\gamma})$, which may be numerically approximated to circumvent computational difficulties. Consequently, we  have a plug-in {estimator}  $\widehat\eta_{g}(z,A,C)$ {of}   $\eta_{g}(z,A,C)$ defined in~\eqref{eq: the-mixture-expression}
% Besides, if we
% % To obtain  the weights $\widehat\omega_g(Z,A,C)$, sampling techniques using the distribution $f(W\mid Z,A,C; \widehat{\gamma})$ can be employed  to 
% % approximate the integral of equation \eqref{eq: principal-score-A-C}. 
% let $\pi_g(A,C)=\pr(G=g\mid A,C)$, then we also find 
and
an {estimator} of $\pi_g(A,C)$ as follows:
$$\widehat{\pi}_{g}(A,C )=\textstyle\sum_{z=0}^{1}\widehat{\omega}_{g}(z,A,C)\pr( z\mid A,C;\widehat\gamma).$$ 
In addition, if the model assumptions in Proposition \ref{prop: idetification-of-principal-stratum} hold, we can further estimate the conditional probabilities given fully observed covariates $\omega_g(Z,X)$ and  $\pi_g(X)$; the estimation details are {relegated} to the supplementary material for space cosiderations. As noted by Theorems \ref{thm: PCE-identification-AC}--\ref{thm: PCE-identification-X}, the estimation of principal causal effects requires different  conditional probabilities of principal strata, depending on which assumptions are imposed. For example, Theorem \ref{thm: PCE-identification-AC} requires $\omega_g(Z,A,C)$, whereas Theorem \ref{thm: PCE-identification-X} requires $\omega_g(Z,X)$. For simplicity, we denote these conditional probabilities  by a unified notation $\omega_g(Z,V)$ with $V=(A,C^\T)^\T$ in Theorems \ref{thm: PCE-identification-AC}--\ref{thm: linear-model-AC} and $V=X$ in Theorem~\ref{thm: PCE-identification-X}. The notations $\eta_g(z,V)$ and $\pi_g(V)$ are equipped with similar meanings.

In the second step, we aim to estimate the parameters $\theta_{z,g}$ for $z=0,1$ and $g=\AT,\NT,\CO$ in the outcome Model~\ref{model: outcome-pce}. To derive an estimator for $\theta_{0,g}$, we observe the following moment constraints by invoking the monotonicity assumption:
\begin{equation} \label{eqn:ee-at}
\begin{aligned}&E\big\{Y-\mu_{0,{\AT}}(X;\theta_{0,{\AT}})\mid Z=0,S=1,X\big\}=0, \\
&E\big\{Y-\textstyle\sum_{g={\CO},{\NT}}\widehat\eta_g(0,V)\mu_{0,g}(X;\theta_{0,g})\mid Z=0,S=0,X\big\}=0.
\end{aligned}
\end{equation} 
We emphasize here that the specifications of $\mu_{z,g}(X;\theta_{z,g})$ in Model~\ref{model: outcome-pce} may not always depend on all the observed covariates $X$ due to identifiability concerns; see also the discussions below Model~\ref{model: outcome-pce}. With the above moment constraints, we can apply the generalized method of moments again to obtain a consistent estimator for $\theta_{0,g}$. The estimation of $\theta_{1,g}$ is similar, because we have another pair of moment constraints:
\begin{equation} \label{eqn:ee-nt}
\begin{aligned}
&E\big\{Y-\mu_{1,{\NT}}(X;\theta_{1,{\NT}})\mid Z=1,S=0,X\big\}=0,\\
&E\big\{Y-\textstyle\sum_{g={\AT},{\CO}}\widehat\eta_g(1,V)\mu_{1,g}(X;\theta_{1,g})\mid Z=1,S=1,X\big\}=0.
\end{aligned}
\end{equation} 
%  We relegate the estimation procedure of $\omega_g(Z,X)$ and  $\pi_g(X)$ considered in Proposition \ref{prop: idetification-of-principal-stratum} to the supplementary material. 
%  It remains to estimate  the {principal causal effects} based on  Theorems \ref{thm: PCE-identification-AC}--\ref{thm: PCE-identification-X}.  Different theorem conditions require different weights. For example, Theorem \ref{thm: PCE-identification-AC} requires $\omega_g(Z,A,C)$, and Theorem \ref{thm: PCE-identification-X} requires $\omega_g(Z,X)$. For   simplicity,  we compress the   weights  as  $ \omega_g(Z,V)$ and $ \pi_g(V)$ with a generic variable $V$.  We hence propose a  unified estimation procedure in full generality that is compatible with the conditions in Theorems  \ref{thm: PCE-identification-AC}--\ref{thm: PCE-identification-X}.
% To derive the   estimator for $\theta_{1,g}$, we find that Model   \ref{model: outcome-pce} correspond to the following moment constraints on the observed data distribution:  
% \begin{equation*} 
% \begin{array}{l}E(Y\mid Z=1,S=0,V)=\mu_{1,{\NT}}(V;\theta_{1,{\NT}}),\\E(Y\mid Z=1,S=1,V)= {\sum_{g={\AT},{\CO}} {{ \omega}_g(1,V)}\mu_{1,g}(V;\theta_{1,g})}/{\{{ \sum_{g={\AT},{\CO}} { \omega}_g(1,V)}\} } .\end{array}
% \end{equation*}  
% We can use GMM with some user-specified function  again to solve for $\widehat\theta_{1,g}$ through the above modelling constraints.  The estimation task for $\widehat\theta_{0,g}$ is familiar,  we omit it for simplicity. 
Finally, in view of \eqref{eq: the-expression-of-PCE}, we can obtain our proposed estimator  for the principal causal effect as follows:
$${\widehat\Delta}_g=\frac{\mathbb{P}_n\big\{\mu_{1,g}(X;{\widehat\theta}_{1,g}){\widehat\pi}_g(V)\big\}}{\mathbb{P}_n\big\{{\widehat\pi}_g(V)\big\}}-\frac{\mathbb{P}_n\big\{\mu_{0,g}(X;{\widehat\theta}_{0,g}){\widehat\pi}_g(V)\big\}}{\mathbb{P}_n\big\{{\widehat\pi}_g(V)\big\}}.$$  
Using  empirical process theories, one can show that the resulting estimator   $\widehat\Delta_{g}$  is consistent and asymptotically normally distributed. 
  
 	  \section{Simulation studies}
 	  \label{sec: simulation} 
We conduct simulation studies to investigate the finite sample performance of the proposed estimators in this section.  
We consider the following data-generating mechanism: 
\begin{itemize}[leftmargin=30pt] 
	\item[(a).] We generate covariates $(A,C)$ from $ \begin{array}{l}(A,C)^\T\sim N\left\{\left(\delta_a,\delta_c\right)^\T,\begin{pmatrix}\sigma_a^2&\rho_1\sigma_{a}  \sigma_{c} \\\rho_1\sigma_{a}  \sigma_{c} &\sigma_c^2\end{pmatrix}\right\}.\end{array}$
	%	\item $W=\mu_{ 0}+\mu_{a}A +\mu_{c}{h(C)}+\varepsilon_w$, where $\varepsilon_w\sim N(0,\sigma_w^2)$. 
%	\item $U=\iota_0+\iota_a A+\iota_c C+\iota_w W+\varepsilon_u$, where $\varepsilon_a\sim N(0,\sigma_u^2)$.
	\item[(b).]  We generate the binary treatment $Z$  from a Bernoulli distribution with $\operatorname{pr}(Z=1 \mid A,C)=\Phi\left(\beta_{0}+\beta_{a} A+\beta_{c} C\right)$. 
	\item[(c).] Given $(Z,A,C)^\T$, we generate $(U,W)$  from the following joint normal distribution
%$$(W,U)\mid Z,A,C\sim N\left\{\begin{pmatrix}\mu_0+\mu_zZ+\mu_aA^2+\mu_cC\\\iota_0+\iota_aZ+\iota_aA^2+\iota_cC\end{pmatrix},\begin{pmatrix}\sigma_w^2&\rho\sigma_w\sigma_u\\\rho\sigma_w\sigma_u&\sigma_u^2\end{pmatrix}\right\}$$
$$(U,W)^\T\mid Z,A,C \sim N\left\{\begin{pmatrix}\iota_0+\iota_zZ+\iota_aA+\iota_{c1}C+\iota_{c2}C^2\\\gamma_0+\gamma_zZ+\gamma_aA+\gamma_{c1}C+\gamma_{c2}C^2\end{pmatrix},\begin{pmatrix}\sigma_u^2&\rho_2\sigma_u\sigma_w\\\rho_2\sigma_u\sigma_w&\sigma_w^2\end{pmatrix}\right\}.$$ 
% where $h(C)$ is some non-linear function of $C$. 
To guarantee $W\indep (Z,A)\mid (U,C)$, we set
%	$$ \dfrac{- \iota_a \sigma_w \rho_2}{\sigma_u} + \gamma_a=0,\; \dfrac{- \iota_z \sigma_w \rho_2  }{\sigma_u} + \gamma_z=0.$$
$  \gamma_a={\iota_a \sigma_w \rho_2}/{\sigma_u}$ and $\gamma_z=\displaystyle  {\iota_z \sigma_w \rho_2}/{\sigma_u }.$ 
For simplicity, we assume that $E(U\mid Z,A,W,C)$ is linear in $C$ by setting
$\gamma_{c2}= \iota_{c2}{\sigma_w}/ { \sigma_u\rho_2}.$ 
%Therefore,  we have $U\mid Z,A,W,C\sim N\left\{m_u(Z,A,W,C),\Sigma_u^2\right\}$, where $m_u(Z,A,W,C)=\nu_0+\nu_zZ+\nu_aA+\nu_cC+\nu_wW$ and 
%$$\begin{array}{c}\nu_0=\left(\iota_0\sigma_w-\beta_0\sigma_u\rho\right)/\sigma_w,\;\nu_z=\left(\iota_z\sigma_w-\beta_z\sigma_u\rho_2\right)/\sigma_w,\;\nu_w=\sigma_u\rho_2/\sigma_w,\\\nu_c=\left(\iota_{c1}\sigma_w-\beta_{c1}\sigma_u\rho_2\right)/\sigma_w,\;\;\nu_a=\left(\iota_a\sigma_w-\beta_a\sigma_u\rho_2\right)/\sigma_w,\;\Sigma_u^2=\sigma_u^2(1-\rho_2^2).\end{array}$$
% Also, $W\mid Z,A,U,C\sim N\left\{m_w(Z,A,W,C),\Sigma_u^2\right\}$, where
%$m_w(Z,A,U,C)=\tau_0+\tau_uU+\tau_{c1}C+\tau_{c2}h(C)$  and%$$\begin{array}{l}E(U\mid A,W,C)\\\;\;\;\;\;\;\;=A\iota_a-\dfrac{A\beta_a\sigma_u\rho_2}{\sigma_w}+{h(C)}\iota_{c2}-\dfrac{{h(C)}\beta_{c2}\sigma_u\rho_2}{\sigma_w}+C\iota_{c1}-\dfrac{C\beta_{c1}\sigma_u\rho_2}{\sigma_w}+\dfrac{W\sigma_u\rho_2}{\sigma_w}+\iota_0-\dfrac{\beta_0\sigma_u\rho_2}{\sigma_w}\\\;\;\;\;\;\;\;=A\iota_a-\dfrac{A\beta_a\sigma_u\rho_2}{\sigma_w}+C\iota_{c1}-\dfrac{C\beta_{c1}\sigma_u\rho_2}{\sigma_w}+\dfrac{W\sigma_u\rho_2}{\sigma_w}+\iota_0-\dfrac{\beta_0\sigma_u\rho_2}{\sigma_w}\end{array}\\$$
%$$\begin{array}{c}\tau_0=\left(\beta_0\sigma_u-\iota_0\sigma_w\rho_2\right)/\sigma_u,\;\tau_{c1}=\left(\beta_{c1}\sigma_u-\iota_{c1}\sigma_w\rho_2\right)/\sigma_u,\\\tau_{c2}=\left(\beta_{c2}\sigma_u-\iota_{c2}\sigma_w\rho_2\right)/\sigma_u,\;\tau_u=\sigma_w\rho_2/\sigma_u,\;\Sigma_w^2=\sigma_w^2(1-\rho_2^2).\end{array}$$ 
	%{\\\blue verification the completeness of $f(W\mid Z=1,A,C)$}
	\item[(d).]
Define $  G^\dag=\zeta_{0}+\zeta_{w}W+\zeta_{u}U+\zeta_{c}C$, and we generate the  principal stratum $G$   from the following ordered probit model:
%$$G=\NT\cdot{\bf 1}(G^{\dag}+\varepsilon\leq 0)+\CO\cdot{\bf 1}\{0<G^{\dag}+\varepsilon\leq \exp({\zeta_1})\}+\AT\cdot{\bf 1}\{\exp({\zeta_1})<G^{\dag}+\varepsilon\},$$ 
%where $\varepsilon\sim N(0,1)$.  Therefore,
%$\mathrm{pr}(G={\NT}\mid U,A,W,C)=\operatorname\Phi\left(-G^{\dag}\right)$, $\mathrm{pr}(G={\CO}\mid U,A,W,C)=\operatorname\Phi\left\{\exp(\zeta_1)-G^{\dag}\right\}-\operatorname\Phi\left(-G^{\dag}\right)$ and $\mathrm{pr}(G={\AT}\mid U,A,W,C)=1-\operatorname\Phi\left\{\exp(\zeta_1)-G^{\dag}\right\}$.
 $$\begin{array}{l}\mathrm{pr}(G={\NT}\mid U,A,W,C)=\operatorname\Phi(-G^{\dag}),\\\mathrm{pr}(G={\CO}\mid U,A,W,C)=\operatorname\Phi\{\exp(\zeta_1)-G^{\dag}\}-\operatorname\Phi(-G^{\dag}),\\\mathrm{pr}(G={\AT}\mid U,A,W,C)= \operatorname\Phi\{G^{\dag}-\exp(\zeta_1)\}.\end{array}$$
	\item[(e).] The
 outcome  $Y$  is finally generated from   the following conditional normal distribution:
$$\begin{array}{c}Y\mid(Z=z,G=g,A,W,C)\sim N( \theta_{z,g,0}+ \theta_{ a}A + \theta_{ w}W+ \theta_{ c}C,\sigma_y^2) .\end{array} $$
\end{itemize} 
 The true values of parameters are set as follows:  
\begin{itemize}[leftmargin=30pt] 
\item[(a).] 
 $ \delta_a=0$, $\delta_c=0$, $\sigma_a=0.5$, $\rho_1=0.5$, $\sigma_c=0.5$.
\item[(b).] $\beta_0 = 0$, $\beta_a =1$, $\beta_c=1$.
 \item[(c).] 
   $ \iota_0=1$, $\iota_z=1$, $\iota_a=1.5$,   $\iota_{c1}=1.5$, $ \iota_{c2}=-0.75$,
   $ \gamma_0=1$,  $  \gamma_z=0.5$,  $ \gamma_a=0.75$,   $\gamma_{c1}=1.5$,  $\gamma_{c2}=-1.5$, $\sigma_u=0.5$,  $\rho_2=0.5$, $\sigma_w=0.5$.
%   , $h(C)=C^2$. 
   \item[(d).]    $\zeta_0=0.5$, $ {\zeta_1}=0$, $ \zeta_w =0.5 $, $ \zeta_c =1$.   Since $\zeta_u$ controls for magnitude of unobserved confounding, we consider 6 different values, i.e.,  $\zeta_u\in\{0, 0.1,0.2,0.3,0.4,0.5\}$.%\wl{Do we need to introduce $\eta$?}
%   ,  where $\eta$ takes values on $0, 0.1, 0.2, 0.3,0.4$ and $0.5$.  
  \item[(e).]
$\theta_{0,{\NT},0}=0$,  $\theta_{0,{\CO},0}=1$,  $\theta_{0,{\AT},0}=2$, 
$\theta_{1,{\NT},0}=2$,  $\theta_{1,{\CO},0}=3$,  $\theta_{1,{\AT},0}=4$, $\theta_c=1$,      $ \sigma_y=0.5$.  
 We consider   4 settings for $(\theta_{a},\theta_{w})$:  
 $ (0,0) $, $ (1,0) $,  $  (0,1) $ and $ (1,1) $, which corresponds to different identifying assumptions.  %so that independence assumption \ref{assumption: indep-of-A} is satisfied; \\
%Scenario  (2): $ (\theta_a,\theta_w)=(1,0) $ under linear model { XXX} or assumption \ref{assumption: indep-of-W}.\\
%Scenario  (3): $ (\theta_a,\theta_w)=(0,1) $ under linear model { XXX} or  assumption \ref{assumption: identification-of-SACE}(c).\\ 
%Scenario  (4): $ (\theta_a,\theta_w)=(1,1) $ so that linear model {XXX}.    
\end{itemize}

% \usepackage{multirow}
%We now make   some comments on the data generation mechanism of the process between treatment and intermediate variables. 
%  We  introduce a parameter $\eta$ to   capture the effect   from $U$ to $G$, different values of  $\eta$ would lead to different proportions of principal strata.  

  \begin{table}[t]
\caption{Simulation studies with bias ($\times 100$), standard error ($\times 100$)  and 95\% coverage probability ($\times 100$) for various settings and sample sizes. %\wl{Illustrate abbreviations Sd, Cp somewhere. Is it necessary to report case of $\zeta_u = 0.2$?}
}\label{tab: result-simulation-0.2}
\centering
\resizebox{0.98\columnwidth}{!}{% 
\begin{threeparttable}
\begin{tabular}{cccrrrrrrrrrrrr}
\toprule[1pt]   $n$& Case  &  \multicolumn{1}{c}{$(\theta_a,\theta_w)$} && \multicolumn{3}{c}{$\Delta_\AT$} &  & \multicolumn{3}{c}{$\Delta_\CO$} &  & \multicolumn{3}{c}{$\Delta_\NT$} \\\addlinespace[0.5mm]\midrule[1pt]\addlinespace[1mm]   &&&&&\multicolumn{9}{c}{$\zeta_u=0.2$}   \\\addlinespace[0.5mm]\cline{5-15} \addlinespace[1mm]\addlinespace[0.5mm]         \addlinespace[0.5mm]    &      &         &   &    \multicolumn{1}{c}{Bias}         &    \multicolumn{1}{c}{Sd}      & \multicolumn{1}{c}{CP}          &  & \multicolumn{1}{c}{Bias}         &    \multicolumn{1}{c}{Sd}      & \multicolumn{1}{c}{CP}      &  & \multicolumn{1}{c}{Bias}         &    \multicolumn{1}{c}{Sd}      & \multicolumn{1}{c}{CP}        \\\addlinespace[0.25mm]\addlinespace[0.25mm] $ 1000$ & (i)&  $(0,0)   $ & & $-$0.6     & 7.5    & 96.2   &    & 2.4      & 39.8   & 95.8   &    & $-$2.8     & 28.6   & 95.0     \\ 
                     &     (ii) & $(1,0)     $  &  & $-$0.5     & 8.0      & 95.8   &    & 0.1      & 46.3   & 96.4   &    & $-$0.7     & 21.9   & 95.8   \\
                &          (iii)   & $(0,1)   $   &  & $-$1.0       & 12.5   & 95.6   &    & 7.8      & 48.5   & 94.0     &    & $-$3.3     & 21.6   & 94.4   \\
                &           (iv)   &$ (1,1)   $   &  & $-$0.2     & 12.6   & 95.8   &    & 1.1      & 50.3   & 96.8   &    & $-$0.9     & 23.4   & 95.2   \\  \addlinespace[2mm]   &         &          &  & \multicolumn{1}{c}{Bias}         &    \multicolumn{1}{c}{Sd}      & \multicolumn{1}{c}{CP}          &  & \multicolumn{1}{c}{Bias}         &    \multicolumn{1}{c}{Sd}      & \multicolumn{1}{c}{CP}      &  & \multicolumn{1}{c}{Bias}         &    \multicolumn{1}{c}{Sd}      & \multicolumn{1}{c}{CP}       \\\addlinespace[0.25mm]
$5000$       
            &(i) &$ (0,0)     $                &  & $-$0.4     & 3.5    & 94.4   &    & 1.6      & 18.9   & 96.6   &    & 0.7      & 12.9   & 94.8   \\ 
         &(ii)   & $(1,0)    $                  &  & $-$0.4     & 3.8    & 94.6   &    & 1.0        & 21.9   & 96.4   &    & 0.3      & 9.6    & 95.8   \\
          &(iii)  & $(0,1)$                       &  & $-$1.2     & 5.3    & 95.8   &    & 4.9      & 20.8   & 96.0     &    & $-$0.1     & 9.5    & 95.6   \\
      &(iv)      & $(1,1) $                     &  & $-$1.0       & 5.3    & 96.4   &    & 3.7      & 21.8   & 96.0     &    & 0.3      & 10.1   & 96.6   \\ \addlinespace[0.5mm]\midrule[1pt] \addlinespace[1mm]     &&&&&\multicolumn{9}{c}{$\zeta_u=0.5$} \\\addlinespace[0.5mm]\cline{5-15} \addlinespace[1mm]\addlinespace[0.5mm]         \addlinespace[0.5mm]  &          &         &  & \multicolumn{1}{c}{Bias}         &    \multicolumn{1}{c}{Sd}      & \multicolumn{1}{c}{CP}          &  & \multicolumn{1}{c}{Bias}         &    \multicolumn{1}{c}{Sd}      & \multicolumn{1}{c}{CP}      &  & \multicolumn{1}{c}{Bias}         &    \multicolumn{1}{c}{Sd}      & \multicolumn{1}{c}{CP}     \\\addlinespace[0.25mm]
$1000$      
          &(i)  & $(0,0)    $                    && $-$0.5                     & 5.5                    & 95.2                   &    & 8.2                      & 44.8                   & 93.6                   &    & $-$6.6                     & 25.0                     & 94.4                   \\
          
    &  (ii)        & $(1,0)    $     &             & $-$0.5                     & 5.9                    & 95.8                   &    & 8.3                      & 51.0                     & 94.8                   &    & $-$3.2                     & 19.0                     & 94.8                   \\
    &(iii)         & $(0,1)  $      &                & $-$0.6                     & 6.8                    & 96.6                   &    & 11.1                     & 47.0                     & 93.8                   &    & $-$3.5                     & 20.9                   & 93.4                   \\
        &(iv)     & $(1,1)   $        &            & $-$0.5                     & 6.8                    & 96.4                   &    & 6.4                      & 51.0                     & 96.2                   &    & $-$2.2                     & 22.4                   & 95.4                   \\ \addlinespace[0.5mm]
            \addlinespace[1mm]    &                  &  & & \multicolumn{1}{c}{Bias}         &    \multicolumn{1}{c}{Sd}      & \multicolumn{1}{c}{CP}          &  & \multicolumn{1}{c}{Bias}         &    \multicolumn{1}{c}{Sd}      & \multicolumn{1}{c}{CP}      &  & \multicolumn{1}{c}{Bias}         &    \multicolumn{1}{c}{Sd}      & \multicolumn{1}{c}{CP}       \\\addlinespace[0.25mm]
$5000$  & (i)       & $(0,0)    $                 &  & $-$0.4                     & 2.6                    & 95.0                     &    & 5.1                      & 21.8                   & 94.4                   &    & $-$0.1                     & 11.8                   & 95.4                   \\
           &  (ii)      & $(1,0)  $                 &  & $-$0.4                     & 2.8                    & 94.6                   &    & 4.8                      & 24.6                   & 94.4                   &    & $-$0.2                     & 8.6                    & 95.6                   \\
          &   (iii)   & $(0,1)    $               &  & $-$0.7                     & 3.2                    & 95.2                   &    & 7.1                      & 21.6                   & 94.0                     &    & $-$0.4                     & 8.8                    & 94.4                   \\
      &       (iv)        & $(1,1) $            &  & $-$0.6                     & 3.2                    & 95.2                   &    & 6.0                        & 23.2                   & 94.2                   &    & 0.0                        & 9.3                    & 95.4                  \\  \addlinespace[0.5mm]\bottomrule[1pt] 
\end{tabular}
\begin{tablenotes} 
 \item Sd: empirical standard error.~ CP: 95\% coverage probability.
\end{tablenotes}
 \end{threeparttable} 
} 
\end{table} 

Under the above data generating mechanism, the  bridge function with the following form
 is compatible with Assumption~\ref{assumption:  bridge-fun} (see supplementary materials for details):
 $$  h(Z,W,C)=\Phi\left\{\alpha_0+\exp( \alpha_1)Z+\alpha_wW+\alpha_{c1}C+\alpha_{c2}C^2\right\}.$$    We 
 thus  model the bridge function in Model \ref{model: bridge-function} with this parametric  form and  specify all correct  parametric models      in   Model \ref{model: distribution-of-models}. 
 We estimate the bridge function by solving the estimating equation \eqref{eq: parameter-bridge-function} with the user-specified functions $B(Z,A,C)=\{1,A,Z,C,C^2\}^\T$. 
	{It is worth pointing out that our data-generating mechanism also satisfies the model assumptions in Proposition \ref{prop: idetification-of-principal-stratum},  and  we can consistently estimate the probabilities  $\omega_g(Z,X)$ and    $\pi_g(X)$   using the  method given in the supplementary material.  } %Specifically, 	to guarantee the linear independence condition  of Proposition  \ref{prop: idetification-of-principal-stratum}, we  employ a quadratic term $h(C)$ as a predictor  to  Model \ref{model: distribution-of-models}(b).}    
% {\red In the appendix, we also considered a more complicated multivariate probit model, which requires  three principal strata link with two  probit models.}
%To examine the performance of our proposed estimators, we consider six estimation scenarios corresponding to different settings of  $(\theta_a,\theta_w)$.  
We investigate the performance of the proposed estimators under various values of $(\theta_a, \theta_w)$, which represent different conditional independence conditions between $(A,W)$ and the outcome $Y$. {For the four different settings in (e), we consider estimation of  principal causal effects with four different correct  parametric forms   in Model \ref{model: outcome-pce}, respectively.} For example, the setting {$(\theta_a, \theta_w)=(0,0)$ implies that Assumption~\ref{assump: negative-control-indep} holds, and we specify the working model $\mu_{z,g}(X;\theta_{zg})=\theta_{z,g,0}+\theta_{ c}C$}; if $(\theta_a, \theta_w)=(1,1)$,  then  the outcome {can} be affected by all covariates{,  and we employ the linear model given in \eqref{eq: linear-model-2}}. For simplicity, we refer to these four different {estimation procedures} as cases (i)--(iv), respectively.

% by using different estimators compatible with the conditions of Theorem \ref{thm: PCE-identification-X}.  
% We refer to the estimation methods corresponding to Theorem \ref{thm: PCE-identification-X} (i)--(iv) in Section \ref{sec: estimation} as Case (i)--Case (iv), respectively.  For example, the parameters  $(\theta_a,\theta_w)=(0,0)$,  which satisfying the condition of Theirem \ref{thm: PCE-identification-X}(i),   {principal causal effects} in this setting can be estimated according to the Case (i).

For each value of $\zeta_u$, we consider sample size $n=1000$  and $n=5000$. 
Table  \ref{tab: result-simulation-0.2} reports the bias,   standard error and coverage probabilities of 95\% confidence intervals averaged across 500 replications with {$\zeta_u=0.2$ and $0.5$.} The corresponding results for other values of $\zeta_u$ are provided in the supplementary material. The results in all the settings are similar. 
 It can be found that our method {has} negligible biases {with} smaller variances as the sample size increases.  
 {Estimators} of $\Delta_{\AT}$ and $\Delta_{\NT}$ are more stable than that of $\Delta_{\CO}$. {This may be because}  the estimation of $\Delta_{\CO}$ requires solving the joint estimating equations~\eqref{eqn:ee-at} and~\eqref{eqn:ee-nt} rather than only one of them. {The} proposed estimators have coverage probabilities close to the nominal level in all scenarios. All these results  {demonstrate the  consistency  of   our proposed estimators.}

% is expected because the estimation of $\Delta_{\CO}$ is more complicated and it requires us to solve one more estimating equation, which results in a more variable estimator. 
% Finally, the proposed estimators have coverage close to the nominal level in all scenarios.  

% latex table generated in R 3.6.3 by xtable 1.8-4 package
% Tue Feb  8 20:20:19 2022
% latex table generated in R 3.6.3 by xtable 1.8-4 package
% Tue Feb  8 20:34:46 2022
% \usepackage{multirow}
\section{Application to Return to Schooling} 
We illustrate our approach by 
% using  an observational study  with non-compliance. 
% We 
reanalyzing the dataset from the National Longitudinal Survey of Young Men \citep{card1993using,Tan2006Regression}. 
%The dataset from the U.S. National Longitudinal Survey of Young Men contains 3010 men with age between 14 and 24 in the year 1966. 
This cohort study includes 3,010 men who were aged $14$-$24$ when first interviewed in 1966, with follow-up surveys continued until 1976.  
 We are interested in estimating the causal effect of education on earnings, which might be confounded by unobserved preferences for {students'} abilities and family costs \citep{KEDAGNI2021}. 

The treatment $Z$ is an indicator of living near a four-year college. Following \citet{Tan2006Regression},
we choose the {educational} experience beyond high school as the intermediate variable $S$. The outcome $Y$ is the log wage in the year 1976, ranging from 4.6 to 7.8.  We consider the average parental education years as the negative control exposure $A$, because parents'  education years {are highly correlated with} whether their children have the chance to live close to a college. We use the intelligence quotient (IQ) scores as the 
negative control intermediate variable $W$,   because IQ {is related to} students' learning abilities{, and} students with higher IQ are more likely to enter college.   
The data set also {includes} the following covariates $C$: race, age, scores on the Knowledge of the World of Work 
% (KWW)
test, a categorical variable indicating whether children living with both parents, single mom, or {step-parents}, and several geographic variables summarizing living areas in the past.
% {some binary covariates  indicating   for whether a subject lived with both parents, single mom and step parents at age 14, and several  geographic variables  that summarizing the living areas in 1966. }%a categorical variable indicating living with both parents, single mom, or both parents, and variables summarizing the living areas in 1966.   
The missing covariates are imputed via the $k$-Nearest Neighbor algorithm with $k=10$ \citep{franzin2017bnstruct}.

Monotonicity is plausible because living near a college would make an individual more likely to receive higher education. Following \citet{jiang2020multiply},  we do not invoke the exclusion restriction assumption that living near a college can affect the earnings only through education. In fact, we can evaluate the validity of this assumption by applying the proposed approach in this paper.
% Our proposed estimation method can be used to evaluate this assumption. 
We employ similar model parameterizations
% of Models \ref{model: bridge-function}-\ref{model: outcome-pce}   considered 
as used
in simulation studies, and our analyses here are also conducted under the cases (i)--(iv) that represent different conditional independence assumptions for the outcome model.
% To guarantee the linear independence condition of Proposition  \ref{prop: idetification-of-principal-stratum}, we  add a  square term of  age as a predictor  to  the conditional expectation   $E(W\mid Z,A,C)$ in Model \ref{model: distribution-of-models}.
%Again we analyze this real data using the 6 scenarios considered in the simulation, specifically, we postulate a linear model for the outcome mean and probit models for the principal score.  

%  in Table \ref{tab: NLSY-res}
%  ,  wherein the 95\% confidence intervals of  $\Delta_g$ for the estimates are obtained by the nonparametric bootstrap method.  
\begin{table}[t]
\centering
\caption{Analysis of the National Longitudinal Survey of Young Men.}
\label{tab: NLSY-res}  
\resizebox{0.78\columnwidth}{!}{% 
\begin{tabular}{cccc}
  \hline  \addlinespace[0.5mm]
 Case&  \multicolumn{1}{c}{$~~~~~~~~~~~~\Delta_{\AT}$}   &  \multicolumn{1}{c}{$~~~~~~~~~~~~\Delta_{\CO}$}  &  \multicolumn{1}{c}{$~~~~~~~~~~~~\Delta_{\NT}$}  \\ 
 \addlinespace[0.25mm]\hline \addlinespace[1mm] 
 
 (i) & $~~~$0.07~ ($-$0.07, $~~~$0.81) & $-$0.68 ~($-$1.85, $-$0.25) & $-$0.86~ ($-$2.83, $-$0.32) \\\addlinespace[0.5mm] 
   
(ii) & $~~~$0.06~ ($-$0.09, $~~~$0.77) & $~~~$0.13 ~($-$0.45, $~~~$0.59) & $~~~$0.02~ ($-$0.56, $~~~$0.16)  \\\addlinespace[0.5mm] 
 (iii) & $-$0.18~ ($-$0.55, $~~~$0.80) & $~~~$0.87~ ($~~~$0.05, $~~~$1.93) & $-$0.07~ ($-$0.91, $~~~$0.23) \\\addlinespace[0.5mm] 
 (iv) & $-$0.18 ~($-$0.56, $~~~$0.78) & $~~~$0.85~ ($~~~$0.00, $~~~$1.88) & $-$0.07~ ($-$0.89, $~~~$0.26)  \\\addlinespace[0.5mm]   \addlinespace[0.25mm]\hline
\end{tabular} }
\end{table}

%In Table \ref{tab: NLSY-res}, we present all estimates of g, along with 95\% confidence intervals for the estimates obtained by the bootstrap method 
%Firstly, we find that the point estimates and 95\% confidence intervals are very similar for Cases (V) and (VI), but not for other cases. From Table \ref{table:simulation-setup}, it can be seen that none of the Cases (I)-(IV)   include the  negative control variable $W$ as a predictor.
Table \ref{tab: NLSY-res} shows the point estimates and their associated 95\% confidence intervals obtained via the nonparametric bootstrap method. We first observe that the results in cases (iii) and (iv) are very {close}. Compared with them, the corresponding results in cases (i) and (ii) are completely different. Because
the outcome model in cases (i) and (ii) do not include the proxy variable $W$ as a predictor,
 the empirical findings may indicate misspecifications of outcome models in these two
 cases. Thus, the results in cases (iii) and (iv), where  the IQ score $W$ is allowed  to directly affect the wage $Y$, are more credible. Based on these results,
% {We hence adopt the results in cases (iii) and (iv), where we allow the IQ score $W$ to directly affect the wage $Y$.} 
 we find %Based on the results from  cases (iii) and (iv), we find 
 that both the 95\% confidence intervals for  $\Delta_{{\AT}}$ and $\Delta_{{\NT}}$ cover zero,  
which implies no significant evidence of violating the exclusion restriction.  The estimate of $\Delta_{
	{\CO}}$ is positive and its corresponding confidence
interval does not cover zero. This implies that education has a significantly positive effect on earnings, which is consistent with previous analyses \citep{Tan2006Regression,jiang2020multiply,KEDAGNI2021}.

%This is consistent with the conclusion in Jiang et al. (2020) 405
%that education has a significantly positive effect on earnings.
%
%
%
%The point  estimate of $\Delta_{
%{\CO}}$   {\lss   suggests 
% that obtaining a college degree increases the average log wage of the compliers by about 0.85. The corresponding confidence interval of $\Delta_{
%{\CO}}$  does not {cover zero}, which is
%in agreement with
%with the previous findings that education has a significantly positive effect on earnings \citep{Tan2006Regression,jiang2020multiply}. This result may be used as policy support for reducing the cost of college education \citep{KEDAGNI2021}.}  % , which support the conclusion that individuals would  benefit from schooling  in previous studies  \citep{jiang2020multiply}. 

% Note that the confidence interval of $\Delta_{{\CO}}$ is the widest, which is consistent  with our simulation studies. However, even so,  the estimator $\Delta_{
% {\CO}}$ still yields positive and statistically significant estimates. 
% This implies that education has a positive effect on earnings, which support the conclusion that individuals would  benefit from schooling  in previous studies  \citep{jiang2020multiply}. 

\section{Discussion}
With the aid of {a pair of negative controls}, we have established identification and estimation of {principal causal effects}    when the treatment and principal strata are confounded by unmeasured variables. The availability of negative control variables is crucial for the proposed approach.
Although it is generally not possible to test the negative control assumptions via observed
data without {additional} assumptions, the existence of such  variables
is practically reasonable in the empirical example presented in this paper and similar
situations where two or more proxies  of unmeasured variables
may be available \citep{Miao2018Identifying,shi2020multiply,miao2020confounding,cui2020semiparametric}.

The proposed methods may be improved or extended in several directions.  First, we consider parametric methods to solve integral equations involved in our estimation procedure. One may also consider nonparametric estimation techniques to obtain the solutions \citep{Newey2003completeness,chen2012estimation,li2021identification}.
   Second, we relax the commonly-used ignorability assumption by allowing unmeasured confounders between the treatment and principal strata, and it is possible to further relax this assumption and consider the setting where {Assumption \ref{assumption: latent ignorability}(ii) fails}.
   Third, our identifiability results rely on the monotonicity assumption which may not hold in some real applications. In principle,  
   one can  conduct
sensitivity analysis to {assess the principal causal effects of violations of monotonicity assumption}
%assess how results would change if the monotonicity assumption
%were violated by some pre-specified amount 
   % where treatment may have negative side-effects on intermediate variable. 
% the approach we develop here may be extended to perform a sensitivity analysis on the monotonicity assumption
\citep{DingPrincipal}.  
 Finally,  it is also of interest to develop doubly robust estimators for the principal causal effects as provided by \citet{cui2020semiparametric} for average treatment effects.
%  we have restricted our attention to mitigate the unmeasured  unmeasured confounding adjustment via negative controls proposed by \cite{Miao2018Identifying}.  It is of great interest to consider the alternative identification results proposed by \citet{cui2020semiparametric}   and further consider robust estimation strategies for {principal causal effects} in combination with Theorem \ref{thm: PCE-identification-AC}. 
 The study of these issues is beyond the scope of this paper and we leave them as future research topics.

					\bibliographystyle{apalike}
					\bibliography{mybib}

% in the presence of unmeasured confounding leveraging a pair of negative controls. 
%  The theoretical results have demonstrated the twofold role of the negative controls:  
%  (I) identification of  proportions of principal strata with bridge function; (II) identification  of  principal causal effects with some conditional independence.   
% Compared to previous studies without unmeasured  confounding \citep{Ding:2011,Jiang2016, Wang2017iv}, in addition to using a  proxy  variable $W$ for the  principal stratum $G$,  our    identification strategies  of {principal causal effects}  also depend on an additional   proxy  variable $A$ for treatment assignment $Z$.   
% When the identification assumption is invalid, we also establish  identification of {principal causal effects} through some  semiparametric and parametric models to mitigate the confounding. 

%				
	\newpage

\appendix   %仅一个附录时用appendix,否则\appendices
{\centering \section*{Supplementary Material}}
 In the supplementary material, we provide proofs of theorems and claims in the main paper. We also provide additional details for the estimation and simulation studies.
\vspace{1mm}
\setcounter{table}{0}   %从0开始编号,显示出来表会A1开始编号
\setcounter{figure}{0}
%定义编号格式,在数字序号前加字符“A"
\renewcommand{\theproposition}{S\arabic{proposition}}
\renewcommand{\thetheorem}{S\arabic{theorem}}
\renewcommand{\theassumption}{S\arabic{assumption}}
\renewcommand{\thesection}{S\arabic{section}}
\renewcommand{\theequation}{S\arabic{equation}}
\renewcommand{\thelemma}{S\arabic{lemma}}
					  \section{Proofs of propositions and theorems}
\subsection{The proof  of expression \eqref{eq: the-expression-of-PCE}} 
\begin{proof} 
By the   Law of Iterated Expectation (LIE), we have
\begin{equation}
    \label{eq: simipilfy-Y}
    \begin{array}{lcl}E(Y_z\mid G=g)&=&E\left\{E(Y_z\mid G=g,X)\mid G=g\right\}\\&=&\int E(Y_z\mid G=g,X)f(x\mid G=g)\operatorname dx\\&=&{\int E(Y_z\mid G=g,X)\mathrm{pr}(G=g\mid X=x)f(x)\operatorname dx}/{\mathrm{pr}(G=g)}\\&=&{E\{E(Y_z\mid G=g,X)\pi_g(X)\}}/{E\{\pi_g(X)\}}.\end{array}
\end{equation}
Given Assumption \ref{assumption: latent ignorability}(ii), we have $$\begin{array}{rcl}E(Y_z\mid G=g,X)&=&E(Y_z\mid Z=z,G=g,X)=E(Y\mid Z=z,G=g,X)=\mu_{z,g}(X).\end{array}$$   Combining these two pieces, we have  
$$E(Y_z\mid G=g)={E\{\mu_{z,g}(X)\pi_g(X)\}}/{E\{\pi_g(X)\}}.$$
\end{proof}
\subsection{Lemmas}
We first prove {the first conclusion in} Theorem \ref{thm: PCE-identification-AC}, which is summarized as the following lemma.
\begin{lemma}\label{lem: negative-control} 
Suppose that Assumptions \ref{assumption: monotonicity}, \ref{assumption: latent ignorability}(i),  \ref{assumption: negative-control-variables} and \ref{assumption:  bridge-fun} hold. Then the {conditional probabilities} of principal strata 
 are  identified by 
\begin{equation*} 
 \begin{array}{lcl}\omega_{{\AT}}(Z,A,C)= E\{h(0, W,C)\mid Z,A,C\},\\
\omega_{{\NT}}(Z,A,C)= 1-E\{h(1, W,C)\mid  Z,A,C\},\\
\omega_{{\CO}}(Z,A,C)=1-\omega_{{\AT}}(Z,A,C)-\omega_{{\NT}}(Z,A,C).\\
 % \int_{-\infty}^{+\infty}h(z, W,C)f(w\mid a,c)dw
% {\blue \operatorname{pr}(S_t=1\mid Z=z ,A,C)=E\{h(t,W,C)\mid Z=z ,A,C\} .}
 \end{array}
\end{equation*}   
\end{lemma}  
 \begin{proof} 
 Given the equality $\pr(S=1\mid Z=z,C,U)=E\{ h(z, W,C) \mid Z=z,C,U)\}$, for any $Z=z^\prime$, we have that
%$$\pr(S=1\mid Z=z,A,C)=E\{ h(z, W,C) \mid Z=z,A,C)\} . $$
%$$\begin{array}{lcl}\operatorname{pr}(S=1\mid Z=z,A,C)&=&\operatorname{pr}(S_z=1\mid Z=z,A,C)\\&=&E\left\{E(S_z\mid Z=z,A,C,U)\mid Z=z,A,C\right\}\\&=&E\left\{E(S_z\mid Z=z,C,U)\mid Z=z,A,C\right\}\\&=&E\left\{E(S\mid Z=z,C,U)\mid Z=z,A,C\right\}\\&=&E\left[E\left\{h(z,W,C)\mid C,U\right\}\mid Z=z,A,C\right]\\&=&E\left[E\left\{h(z,W,C)\mid Z=z,A,C,U\right\}\mid Z=z,A,C\right]\\&=&E\left\{h(z,W,C)\mid Z=z,A,C\right\}\end{array}$$
%We show a stronger statement:
\begin{equation}
\label{eq: potential-bridge-result}
    \begin{array}{lcl}\operatorname{pr}(S_z=1\mid Z=z^{\prime},A,C) &=&E\left\{E(S_z\mid Z=z^{\prime},A,C,U)\mid Z=z^{\prime},A,C\right\}\\&=&E\left\{E(S_z\mid Z=z,C,U)\mid Z=z^{\prime},A,C\right\}\\&=&E\left\{E(S\mid Z=z,C,U)\mid Z=z^{\prime},A,C\right\}\\&=&E\left\{E\{h(z,W,C)\mid C,U\}\mid Z=z^{\prime},A,C\right\}\\&=&E\left\{E\{h(z,W,C)\mid Z=z^{\prime},A,C,U\}\mid Z=z^{\prime},A,C\right\}\\&=&E\{h(z,W,C)\mid Z=z^{\prime},A,C\}.\end{array}
\end{equation}%Then,
%$$\begin{array}{lcl}\operatorname{pr}(S_z=1\mid Z=z^\prime,A,C)&=&E\left\{E(S_z\mid Z=z^\prime,A,C,U)\mid Z=z^\prime,A,C\right\}\\&=&E\left\{E(S_z\mid Z=z^\prime,A,C,U)\mid Z=z^\prime,A,C\right\}\\&=&E\left\{E(S_z\mid Z=z,C,U)\mid Z=z^\prime,A,C\right\}\\&=&E\left\{E(S\mid Z=z,C,U)\mid Z=z^\prime,A,C\right\}\\&=&E\left\{E\{h(z,W,C)\mid C,U\}\mid Z=z^\prime,A,C\right\}\\&=&E\left\{E\{h(z,W,C)\mid Z=z^\prime,A,C,U\}\mid Z=z^\prime,A,C\right\}\\&=&E\{h(z,W,C)\mid Z=z^\prime,A,C\}\end{array}$$
 where the first and final  equalities are due to LIE, the second and fifth equalities are  due to  Assumption \ref{assumption: negative-control-variables}, the third equality is due to consistency, the forth equality is due to Assumption \ref{assumption:  bridge-fun}.

 Given monotonicity assumption \ref{assumption: monotonicity},  we have the equivalence of $\{G={\AT}\}$ and $\{S_0=1\}$ as well as  $\{G={\NT}\}$ and $\{S_1=0\}$, namely \begin{equation*}
 \begin{array}{lcl}\operatorname{pr}(G={\AT}\mid Z,A,C)= \operatorname{pr}(S_0=1\mid Z, A,C)=E\{h(0, W,C)\mid Z,A,C\},\\
\operatorname{pr}(G={\NT}\mid Z,A,C)=\operatorname{pr}(S_1=0\mid Z,A,C)= 1-E\{h(1, W,C)\mid Z,A,C\},\\
\operatorname{pr}(G={\CO}\mid Z,A,C)=  E\{h(1, W,C)\mid Z,A,C\}- E\{h(0, W,C)\mid Z,A,C\}.
 % \int_{-\infty}^{+\infty}h(z, W,C)f(w\mid a,c)dw
 \end{array}
\end{equation*} 
%$$\begin{array}{lcl}\operatorname{pr}(S=1\mid Z=z,A,C)&=&E\left[E\left\{I(S=1)\mid U,Z=z,A,C\right\}\mid A,C,Z=z\right]\\&=&E\left[E\left\{I(S=1)\mid U,Z=z,C\right\}\mid A,C,Z=z\right]\\&=&E\left[E\left\{h(z, W,C)\mid U,Z=z,C\right\}\mid A,C,Z=z\right]\\&=&E\left[E\left\{h(z, W,C)\mid U,Z=z,A,C\right\}\mid A,C,Z=z\right]\\&=&E\left\{h(z, W,C)\mid A,C,Z=z\right\}\end{array}$$ 
 \end{proof} 
 { \begin{lemma}
 \label{lem: complete}
 Suppose that Assumptions  \ref{assumption: latent ignorability}(i),  \ref{assumption: negative-control-variables} and \ref{assumption:  bridge-fun} hold. We also assume that for any $c,z$ and square-integrable function $g$,  $E\{g(U)\mid Z=z,A,C=c\}=0$ almost surely if and only if $g(U)=0$ almost surely. Then any function $h$  satisfying
\begin{equation}
    \label{eq: obs-bridge-fun}
    \pr(S=1\mid Z,A,C)=E\{h(Z,W,C)\mid Z,A,C\}
\end{equation}
is also a valid outcome bridge function in  Assumption \ref{assumption:  bridge-fun}.
 \end{lemma}
 \begin{proof} 
 {Given \eqref{eq: obs-bridge-fun}, for any $c,z$, we have that \begin{align*}
0&=E\{S-h(z,W,c)\mid Z=z,A,C=c\}\\&=E\left[E\{S-h(z,W,c)\mid Z=z,A,C=c,U\}\mid Z=z,A,C=c\right]  \\&=E\left[E\{S-h(z,W,c)\mid Z=z, C=c,U\}\mid Z=z,A,C=c\right],  
 \end{align*}
 where the second equality is due to the LIE, and the last equality is due to Assumption \ref{assumption: negative-control-variables}.}
Given the completeness condition, we have that 
$$E\{S-h(z,W,c)\mid Z=z,C=c,U\}=0,$$
which indicates that  any function $h(Z,W,C)$ that solves equation \eqref{eq: obs-bridge-fun} also satisfies Assumption \ref{assumption:  bridge-fun}.
 \end{proof}
 }
  
\subsection{The proof of Theorems  \ref{thm: PCE-identification-AC}-\ref{thm: linear-model-AC}} 
 We next prove {the second conclusion in Theorem  \ref{thm: PCE-identification-AC}} and Theorem  \ref{thm: linear-model-AC}.
	\begin{proof}  
		Given the conditions in Lemma~\ref{lem: negative-control}, we know that the weights $\omega_g(Z,A,C)$ and $\eta_g(Z,A,C)$ are identifiable for all $g$.  
 Under monotonicity assumption, the causal estimands  {$\mu_{1,\NT}( X)$ and $\mu_{0,\AT}(X)$} can be identified by
\begin{equation*}
\mu_{0,\AT}( X)=E(Y \mid Z=0, S=1, X),\;\mu_{1,\NT}( X)=E(Y \mid Z=1, S=0, X).
\end{equation*} 
We next show that $\mu_{1,\AT}(X)$ and $\mu_{1,\CO}(X)$  are also identifiable. {For simplicity}, we omit the proof of $\mu_{0,\NT}(X)$ and $\mu_{0,\CO}(X)$. 
Applying LIE to get
\begin{equation} 
\label{eq: LIE-form}\begin{array}{l}E(Y\mid Z=1,S=1,A,C) =\eta_{{\AT}}(1,A,C)\mu_{1,{\AT}}(A,C)+\eta_{{\CO}}(1,A,C)\mu_{1,{\CO}}(A,C).\end{array}
\end{equation}
\begin{enumerate}
    \item Given additional Assumptions \ref{assumption: latent ignorability}(ii) and \ref{assump: negative-control-indep}, we have 
$$\mu_{1,g}(A,C)=E(Y\mid Z=1,G=g,A,C)=E(Y\mid Z=1,G=g,C)=\mu_{1,g}(C).$$
Therefore, \eqref{eq: LIE-form} can be simplified as $$\begin{array}{l}E(Y\mid Z=1,S=1,A,C)={\eta_{{\AT}}(1,A,C)\mu_{1,\AT}(C)+\eta_{{\CO}}(1,A,C)\mu_{1,\CO}(C)} .\end{array}$$
For any $C=c$,   if $\{\eta_{{\AT}}(1,A,c),\eta_{{\CO}}(1,A,c)\}^\T$ is linearly independent, we can identify $\mu_{1,\AT}( C)$ and $\mu_{1,\CO}( C)$. Thus, $\mu_{1,\AT}( A,C)$ and $\mu_{1,\CO}(A, C)$ are also identifiable.
    \item Given the conditions in Theorem \ref{thm: linear-model-AC}, \eqref{eq: LIE-form} can be simplified as 
    $$\begin{array}{l}E(Y\mid Z=1,S=1,A,C) ={\theta_{1,\AT,0}\eta_{{\AT}}(1,A,C)+\theta_{1,\CO,0}\eta_{{\CO}}(1,A,C)} +\theta_aA+\theta_cC,\end{array}$$
   If the functions   
$
   \left\{ \eta_{ {\AT}}(1,A,C),\eta_{ {\CO}}(1,A,C),A,C\right\}^{\mathrm{T}}$   
is linearly independent, we then can identify $\theta_{1,\CO}$, $\theta_{1,\AT}$, $\theta_a$ and $\theta_c$. Thus, $\mu_{1,\AT}(A ,C)$ and $\mu_{1,\CO}(A,C)$ are also identifiable.
\end{enumerate}
Given the identifiability of  $\mu_{1,\AT}(X)$ and $\mu_{1,\CO}(X)$, we show that $\mu_{1,\AT}$ and $\mu_{1,\CO}$ can be identified from the view   of  \eqref{eq: simipilfy-Y}. {Specifically}, we have 
$$  \begin{array}{lcl}E(Y_1\mid G=g)&=&E\left\{E(Y_1\mid G=g,A,C)\mid G=g\right\}\\&=&\int E(Y_1\mid G=g,A,C)f(a,c\mid G=g)\operatorname dx\\&=&{\int E(Y_1\mid G=g,a,c)\mathrm{pr}(G=g\mid A=a,C=c)f(a,c)\operatorname da\operatorname dc}/{\mathrm{pr}(G=g)}\\&=&{E\{E(Y_1\mid G=g,A,C)\pi_g(A,C)\}}/{E\{\pi_g(A,C)\}} \end{array}$$
for $ g\in\{\AT,\CO\}.$   Moreover, given Assumptions \ref{assumption: latent ignorability}(ii) and {\ref{assump: negative-control-indep} or \ref{assumption: indep-of-W}}, we have $$   E(Y_1\mid  G=g,A,C)=E(Y\mid Z=1,G=g,A,C)=\mu_{1,g}(A,C) $$ 
Therefore,   for $g\in\{\AT,\CO\}$, we have
\begin{equation*}
    \label{eq: expression-A-C}
    \mu_{1,g}=\dfrac{E\{\mu_{1,g}(A,C)\pi_g(A,C)\}}{E\{\pi_g(A,C)\}}.
\end{equation*}
Similarly, $\mu_{0,\NT}$ and $\mu_{0,\CO}$ are also identifiable.
	\end{proof}

 \subsection{The identifiability of \eqref{eq: linear-model-2}}
 In this section, we {show} the  identifiability of \eqref{eq: linear-model-2}.
 \begin{proof}
 
By the   {LIE}, we have
\begin{equation*}
    \label{eq: LIE-linear-2}
\begin{aligned}E(Y\mid & Z=1,S=1,A,C)\\\;\;\;\;\;\;\;\;\;\;&=\eta_{{\AT}}(1,A,C)\mu_{1,{\AT}}(A,C)+\eta_{{\CO}}(1,A,C)\mu_{1,{\CO}}(X)\\\;\;\;\;\;\;\;\;\;\;&=\eta_{{\AT}}(1,A,C)E\{\mu_{1,{\AT}}(X)\mid Z=1,G={\AT},A,C\}\\\;\;\;\;\;\;\;\;\;\;\;\;\;\;\;&\;\;\;\;\;+\eta_{{\CO}}(1,A,C)E\{\mu_{1,{\CO}}(X)\mid Z=1,G={\CO},A,C\}\\\;\;\;\;\;\;\;\;\;\;&=\eta_{{\AT}}(1,A,C)\theta_{1,{\AT},0}+\eta_{{\CO}}(1,A,C)\theta_{1,{\CO},0}+\theta_cC+\theta_aA\\\;\;\;\;\;\;\;\;\;\;\;\;\;\;\;&\;\;\;\;\;+\theta_wE(W\mid Z=1,G={\AT},A,C)\eta_{{\AT}}(1,A,C)\\\;\;\;\;\;\;\;\;\;\;\;\;\;\;\;&\;\;\;\;\;+\theta_wE(W\mid Z=1,G={\CO},A,C)\eta_{{\CO}}(1,A,C)\\\;\;\;\;\;\;\;\;\;\;&=\eta_{{\AT}}(1,A,C)\theta_{1,{\AT},0}+\eta_{{\CO}}(1,A,C)\theta_{1,{\CO},0}+\theta_cC+\theta_aA\\\;\;\;\;\;\;\;\;\;\;\;\;\;\;\;&\;\;\;\;\;+\theta_wE(W\mid Z=1,S=1,A,C).~~~~~~~~~~~~~~~~~~~~~~~~~
\end{aligned}
\end{equation*}

 	Given the conditions in Lemma \ref{lem: negative-control}, we know that the proportions of principal strata $\eta_g(Z,A,C)$ for all $g$ are identifiable.  If the functions   
$$
   \left\{ \eta_{ {\AT}}(1,A,C),\eta_{ {\CO}}(1,A,C),A,C,E(W\mid Z=1,S=1,A,C)\right\}^{\mathrm{T}}$$   
is linearly independen, we then can identify $\theta_{1,\CO,0}$, $\theta_{1,\AT,0}$, $\theta_a$, $\theta_c$ and $\theta_w$. Subsequently, $\mu_{1,\AT}(X)$ and $\mu_{1,\CO}(X)$ are also identifiable.
Similarly, $\mu_{0,\NT}(X)$ and $\mu_{0,\CO}(X)$ are also identifiable. \end{proof}
 
\subsection{The proof of Proposition \ref{prop: idetification-of-principal-stratum}}
{ \begin{proof} 
First, we can  identify ${m}(Z,A,C)$ and $\sigma_w^2$ from the observed data based on normal distribution. Next, we consider the identifiablity of ${\psi_1}$, $\psi_0$, $\psi_a$, $\psi_c$ and $\psi_w$. 
Given equation \eqref{eq: ordered-probit-model} in main text, we have that
 $$\begin{array}{l}\mathrm{pr}(G={\NT}\mid Z,X)=\operatorname\Phi\left(-G^{\ast}\right),\\\mathrm{pr}(G={\CO}\mid Z,X)=\operatorname\Phi\left\{\exp(\psi_1)-G^{\ast}\right\}-\operatorname\Phi\left(-G^{\ast}\right),\\\mathrm{pr}(G={\AT}\mid Z,X)=\operatorname\Phi\left\{G^{\ast}-\exp(\psi_1)\right\},\end{array}$$
where $G^{\ast}=\psi_{0}+\psi_{z}Z+\psi_{w}W+\psi_{a}A+\psi_{c}C$. 
The above equalities indicate   \eqref{eq: equiv-weight} holds. According to Lemma \ref{lem: negative-control}, we can nonparametriclly identify  the  distribution $\pr(S_t=1\mid Z, A,C)$. Also,
\begin{equation*}
    \label{eq: proof-prop1}
    \begin{aligned} &\int\mathrm{pr}(S_t=1\mid Z,X)f(W\mid Z,A,C)\operatorname dW\\=&\int\Phi\left\{\psi_0+\exp({\psi_1})(t-1)+\psi_zZ+\psi_aA+\psi_cC+\psi_wW\right\}f(W\mid Z,A,C)\operatorname dW\\=&\Phi\left\{\dfrac{\psi_0+\exp({\psi_1})(t-1)+\psi_zZ+\psi_aA+\psi_cC+\psi_wE(W\mid Z,A,C)}{\sqrt{1+\psi_w^2\sigma_w^2}}\right\}\\=&\Phi\left\{\dfrac{\psi_0+\exp({\psi_1})(t-1)+\psi_zZ+\psi_aA+\psi_cC+\psi_w{m}(Z,A,C)}{\sqrt{1+\psi_w^2\sigma_w^2}}\right\}\\=&\mathrm{pr}(S_t=1\mid Z,A,C).
    \end{aligned}
\end{equation*}
Since $\{1,Z,A,C,E(W\mid Z, A,C)\}^\T$  are linearly independent,  then   we can use probit regression to identify all the parameters.
\end{proof}}
\subsection{The proof of Theorem \ref{thm: PCE-identification-X}}
\begin{proof}

Given the conditions in Proposition \ref{prop: idetification-of-principal-stratum}, we know that the proportions of principal strata $\omega_g(Z,X)$ for all $g$ are identifiable. Also, under Assumption \ref{assumption: monotonicity}, we know that the causal estimands 
$\mu_{0,\AT}(X)$ or $\mu_{0,\NT}( X)$  can be identified by
\begin{equation*}
\mu_{0,\AT}( X)=E(Y \mid Z=0, S=1, X),\;\mu_{1,\NT}( X)=E(Y \mid Z=1, S=0, X).
\end{equation*} 
We then show that $\mu_{1,\AT}(X)$ and $\mu_{1,\CO}(X)$  are also identifiable. We omit the proof of $\mu_{0,\NT}(X)$ and $\mu_{0,\CO}(X)$ due to the similarity. 
Applying {LIE} to get
\begin{equation} 
\label{eq: LIE-form-2}
\begin{array}{l}E(Y\mid Z=1,S=1,X) ={\eta_{{\AT}}(1,X)\mu_{1,\AT}(X)+\eta_{{\CO}}(1,X)\mu_{1,\CO}(X)} .\end{array}
\end{equation}

\begin{enumerate}
    \item Given the conditions in Theorem \ref{thm: PCE-identification-X}(i), we have 
$$\mu_{1,g}(X)=E(Y\mid Z=1,G=g,X)=E(Y\mid Z=1,G=g,C)=\mu_{1,g}(C).$$
Therefore, equation \eqref{eq: LIE-form-2} can be simplified as $$\begin{aligned}E(Y\mid Z=1,S=1,X)&={\eta_{{\AT}}(1,X)\mu_{1,\AT}(C)+\eta_{{\CO}}(1,X)\mu_{1,\CO}( C)}. \end{aligned}$$
For any $C=c$, the vector function $\{\eta_{{\AT}}(1,A,W,c),\eta_{{\CO}}(1,A,W,c)\}^\T$ is linearly independent, we then can identify $\mu_{1,{\AT} }(C)$ and $\mu_{1,\CO}(C)$. Thus, $\mu_{1,\AT}(X)$ and $ \mu_{1,\CO}(X)$ are also identifiable with additional Assumption \ref{assump: negative-control-indep}.
\item Given the conditions in Theorem \ref{thm: PCE-identification-X}(ii), we have 
$$\mu_{1,g}(X)=E(Y\mid Z=1,G=g,X)=E(Y\mid Z=1,G=g,A,C)=\mu_{1,g}(A,C).$$
Therefore, equation \eqref{eq: LIE-form-2} can be simplified as $$\begin{aligned}E(Y\mid Z=1,S=1,X)&={\eta_{{\AT}}(1,X)\mu_{1,\AT}(A,C)+\eta_{{\CO}}(1,X)\mu_{1,\CO}( A,C)} .\end{aligned}$$
 Given any $(A,C)=(a,c)$, the function  $ \{ \eta_{{\AT}}(1,a,W,c),\eta_{{\CO}}(1,a,W,c)\}^\T$
	    are linearly independent, we then can identify $\mu_{1,\AT}(A,C)$ and $ \mu_{1,\CO}(A,C)$. Thus, $\mu_{1,\AT}(X)$ and $\mu_{1,\CO}(X)$ are also identifiable with additional Assumption \ref{assumption: indep-of-W}.
\item Given the conditions in Theorem \ref{thm: PCE-identification-X}(iii), we have 
$$\mu_{1,g}(X)=E(Y\mid Z=1,G=g,X)=E(Y\mid Z=1,G=g,W,C)=\mu_{1,g}(W,C).$$
Therefore, equation \eqref{eq: LIE-form-2} can be simplified as  $$\begin{aligned}E(Y\mid Z=1,S=1,X)&={\eta_{{\AT}}(1,X)\mu_{1,\AT}(W,C)+\eta_{{\CO}}(1,X)\mu_{1,\CO}(W,C)}.\end{aligned}$$
 Given any $(W,C)=(w,c)$, the function  $ \{\eta_{{\CO}}(1,A,w,c),\eta_{{\AT}}(1,A,w,c)\}^\T$
	    is linearly independent, we then can identify $\mu_{1,{\AT}}(W,C)$ and $\mu_{1,\CO}(W,C)$. Thus, $\mu_{1,\AT}(X)$ and $\mu_{1,\CO}(X)$ are also identifiable  with additional Assumption \ref{assumption: indep-of-A}.
    \item Given the conditions in  Theorem \ref{thm: PCE-identification-X}(iv), equation \eqref{eq: LIE-form-2} can be simplified as 
     $$\begin{array}{l}E(Y\mid Z=1,S=1,X)\\~~~~~~~={\theta_{1,\AT,0}\eta_{{\AT}}(1,X)+\theta_{1,\CO,0}\eta_{{\CO}}(1,X)}+\theta_aA+\theta_cC+\theta_wW,\end{array}$$
   if the function   $ \left\{ \eta_{ {\CO}}(1,X),\eta_{ {\AT}}(1,X),X\right\}^{\mathrm{T}}  
$ 
is linearly independent, we then can identify $\theta_{1,\CO,0}$, $\theta_{1,\AT,0}$, $\theta_a$, $\theta_w$  and $\theta_c$. Thus, $\mu_{1,\AT}(X)$ and $\mu_{1,\CO}(X)$ are also identifiable.
\end{enumerate}
Given the identifiability of  $\mu_{1,\AT}(X)$ and $\mu_{1,\CO}(X)$, we can  identify $\mu_{1,\AT}$ and $\mu_{1,\CO}$  from   \eqref{eq: simipilfy-Y}. Similarly, we can identify $\mu_{0,\NT}$ and $\mu_{0,\CO}$.

\end{proof} 
\section{Estimation details}
\subsection{Estimation details about Model \ref{model: bridge-function}} 
From \eqref{eq: potential-bridge-result}, we know that 
$$    \begin{array}{lcl}\operatorname{pr}(S_1=1\mid Z=z^{\prime},A,C) &=&E\{h(1,W,C)\mid Z=z^{\prime},A,C\}\\
 &\geq&E\{h(0,W,C)\mid Z=z^{\prime},A,C\}\\
 &=&\operatorname{pr}(S_0=1\mid Z=z^{\prime},A,C),\end{array}$$
which is compatible with the monotonicity assumption.
\subsection{Estimation details about Proposition \ref{prop: idetification-of-principal-stratum}}
If we  assume 
the conditions in Proposition \ref{prop: idetification-of-principal-stratum} hold, that is,  {the conditional distribution   $f(W\mid Z,A,C)$} in Model \ref{model: distribution-of-models} is normal distributed and  equation \eqref{eq: equiv-weight} holds,  
we can  further estimate the weights $\omega_g(Z,X)$ and $\pi_g( X)$.   In order to ensure the linearly independent condition  in Proposition \ref{prop: idetification-of-principal-stratum}, we suggest adding higher-order polynomial, square, or interaction terms  to  the conditional expectation    $m( Z,A,C)$ of the {conditional} distribution  $f(W\mid Z,A,C)$, especially in the case of linear regression.
After obtaining  $\widehat{\omega}_{g}(Z,A,C)$ as shown in the main text, there are some approaches to estimate the parameters in    equation \eqref{eq: equiv-weight}. For example, we can use  GMM again to solve  $\widehat{\psi}$ through  the  following    moment  constraints, 
 $$\begin{array}{rcl}~~~~\widehat{\omega}_{\AT}(Z,A,C)=E\{\mathrm{pr}\left(S_0=1\mid Z,X;\psi\right)\mid Z,A,C;\widehat\gamma \},~~~~~~~~~\\~~~~\widehat{\omega}_{\CO}(Z,A,C)+\widehat{\omega}_{\AT}(Z,A,C)=E\{\mathrm{pr}\left(S_1=1\mid Z,X;\psi\right)\mid Z,A,C;\widehat\gamma \},~~~~~~~~~\end{array}$$or  we can directly derive the specific form of the right hand of the above estimating equation,  and then  use Probit  regression to solve for $\widehat{\psi}$. Under monotonicity assumption \ref{assumption: monotonicity}, we can estimate $\omega_g(Z,X)$ by plugging $\widehat{\psi}$ into  equation \eqref{eq: equiv-weight}:
 \begin{equation*}
\label{eq:omega_g_X} \begin{array}{l}{\widehat\omega}_{{\AT}}(Z,X)=\mathrm{pr}(S_0=1\mid Z,X;\widehat\psi),\;{\widehat\omega}_{{\NT}}(Z,X)=1-\mathrm{pr}(S_1=1\mid Z,X;\widehat\psi),\\{\widehat\omega}_{{\CO}}(Z,X)=\mathrm{pr}(S_1=1\mid Z,X;\widehat\psi)-\mathrm{pr}(S_0=1\mid Z,X;\widehat\psi).\end{array}
\end{equation*}
We then find the estimate of ${\pi}_g(X)$ as follows:
\begin{equation*}
\label{eq:pi_g_X}
{\widehat\pi}_g(X)=\dfrac{\sum_{z=0}^1{\widehat\omega}_g(z,X)\mathrm{pr}(z\mid A,C;\widehat\beta)f(W\mid Z=z,A,C;\widehat\gamma)}{\sum_{z=0}^1\mathrm{pr}(z\mid A,C;\widehat\beta)f(W\mid Z=z,A,C;\widehat\gamma)}.
\end{equation*}
 \section{Simulation details}
\subsection{Simulation details about bridge function in Section \ref{sec: simulation}}
\begin{itemize}
    \item[1. ] We first present the specific form of the conditional density function  $f(W\mid U,C)$. 
Since  $W \mid Z,A,C,U \sim N\{q_w( Z,A,C,U),\Sigma_w^2\}$, where $\Sigma_w^2=\sigma_w^2(1-\rho^2_2)$ and the specific form of  $q_w(Z,A,C,U)$ is
%$$\mu_w(Z,A,C,U)=- \dfrac{Z \iota_z \sigma_w \rho}{\sigma_u} + Z \mu_z - \dfrac{A \iota_a \sigma_w \rho}{\sigma_u} + A \mu_a - \dfrac{C \iota_c \sigma_w \rho}{\sigma_u} + C \mu_c + \dfrac{U \sigma_w \rho}{\sigma_u} - \dfrac{\iota_0 \sigma_w \rho}{\sigma_u} + \mu_0$$ 
\begin{align*}
&~~~~~q_w(Z,A,C,U)=\tau_0+\tau_zZ+\tau_a A+\tau_{c1}C+\tau_{c2}C^2+\tau_uU,\\
  & \tau_0= \dfrac{\gamma_0\sigma_u-\iota_0\sigma_w\rho_2}{\sigma_u},~\tau_z= \frac{\gamma_z \sigma_u - \iota_z \sigma_w \rho_2}{\sigma_u}=0,~\tau_a= \frac{\gamma_a \sigma_u - \iota_a \sigma_w \rho_2}{\sigma_u}=0,\\&\tau_{c1}= \frac{\gamma_{c1} \sigma_u - \iota_{c1} \sigma_w \rho_2}{\sigma_u},\;\tau_{c2}= \frac{\gamma_{c2} \sigma_u - \iota_{c2} \sigma_w \rho_2}{\sigma_u},\;\tau_u=\displaystyle \frac{\sigma_w \rho_2}{\sigma_u}.
\end{align*} 
	\item[2.] We next present the specific form of the latent distribution $\mathrm{pr}(S_z=1\mid U,C)$ is as follows,
		\begin{align*}
		    &\operatorname{pr}(S=1\mid Z=1,U,C)\\&~~~~~~=\operatorname{pr}(S_1=1\mid U,C)\\&~~~~~~=E\{\operatorname{pr}(S_1=1\mid U,A,W,C)\mid U,C\}\\&~~~~~~=E\{\operatorname{pr}(S_1=1\mid U,W,C)\mid U,C\}\\&~~~~~~=E\{\Phi(\zeta_0+\zeta_wW+\zeta_uU+\zeta_cC)\mid U,C\}\\&~~~~~~=\operatorname\Phi\left\{\dfrac{\zeta_0+\zeta_wE(W\mid U,C)+\zeta_uU+\zeta_cC}{\sqrt{1+\zeta_w^2\Sigma_w^2}}\right\}\\&~~~~~~=\operatorname\Phi\left\{\dfrac{\zeta_0+\zeta_w(\tau_0+\tau_uU+\tau_{c1}C+\tau_{c2}{C^2})+\zeta_uU+\zeta_cC}{\sqrt{1+\zeta_w^2\Sigma_w^2}}\right\}\\&~~~~~~=\operatorname\Phi\left\{\dfrac{\zeta_0+\zeta_w\tau_0+\left(\zeta_u+\zeta_w\tau_u\right)U+\left(\zeta_c+\zeta_w\tau_{c1}\right)C+\zeta_w\tau_{c2}{C^2}}{\sqrt{1+\zeta_w^2\Sigma_w^2}}\right\}.
		    \end{align*} 
		    
		\begin{align*}
		    &\operatorname{pr}(S=1\mid Z=0,U,C)\\&~~~~~~=\operatorname{pr}(S_0=1\mid U,C)\\&~~~~~~=E\{\operatorname{pr}(G=\mathrm{AT}\mid U,A,W,C)\mid U,C\}\\&~~~~~~=E\{\operatorname{pr}(G=\mathrm{AT}\mid U,W,C)\mid U,C\}\\&~~~~~~=E\{\Phi(\zeta_0-\exp(\zeta_1)+\zeta_wW+\zeta_uU+\zeta_cC)\mid U,C\}\\&~~~~~~=\operatorname\Phi\left\{\dfrac{\zeta_0-\exp(\zeta_1)+\zeta_wE(W\mid U,C)+\zeta_uU+\zeta_cC}{\sqrt{1+\zeta_w^2\Sigma_w^2}}\right\}\\&~~~~~~=\operatorname\Phi\left\{\dfrac{\zeta_0+\zeta_w\tau_0-\exp(\zeta_1)+\left(\zeta_u+\zeta_w\tau_u\right)U+\left(\zeta_c+\zeta_w\tau_{c1}\right)C+\zeta_w\tau_{c2}{C^2}}{\sqrt{1+\zeta_w^2\Sigma_w^2}}\right\}.	\end{align*} 
		And thus, 
		\begin{equation}
		    \label{eq: density-s-UC} 
		    \begin{aligned}
		    \operatorname{pr}(S_t=1\mid U,C)&= \operatorname\Phi\left\{\dfrac{\zeta_0+\zeta_w\tau_0-\exp(\zeta_1)+\exp(\zeta_1)t+\left(\zeta_u+\zeta_w\tau_u\right)U}{\sqrt{1+\zeta_w^2\Sigma_w^2}}\right.\\&\;\;\;\;\;\;\;\;\;\;\;\;\;\;\;\;\;\;\;\;\;\;\;\;\;\;\;\left.+\dfrac{\left(\zeta_c+\zeta_w\tau_{c1}\right)C+\zeta_w\tau_{c2}{C^2}}{\sqrt{1+\zeta_w^2\Sigma_w^2}}\right\}.
		    \end{aligned}
		\end{equation}
		\item[3.]  We finally  verify  the  Assumption \ref{assumption:  bridge-fun}, that is, our  data generating mechanism is compatible with the condition $ 
\pr(S=1\mid Z=z,C,U)=E\{ h(z, W,C) \mid  C,U \}$ holds for all $z$, where
\begin{align*}
    h(t,W,C)=\Phi\left\{\alpha_0+\exp( \alpha_1)t+\alpha_wW+\alpha_{c1}C+\alpha_{c2}{C^2}\right\}.
\end{align*}
We verify as follows:
\begin{equation}
    \label{eq: birdge-uc}
 \begin{aligned}
     E\{&h(t,W,C)\mid U,C\}\\=&E\left[\Phi\left\{\alpha_0+\exp( \alpha_1)t+\alpha_wW+\alpha_{c1}C+\alpha_{c2}{C^2}\right\}\mid U,C\right]\\=&\operatorname\Phi\left\{\dfrac{\alpha_0+\exp( \alpha_1)t+\alpha_wE(W\mid U,C)+\alpha_{c1}C+\alpha_{c2}{C^2}}{\sqrt{1+\alpha_w^2\Sigma_w^2}}\right\}\\=&\operatorname\Phi\left\{\dfrac{\alpha_0+\alpha_w\tau_0+\exp( \alpha_1)t+\alpha_w\tau_uU+\left(\alpha_w\tau_{c1}+\alpha_{c1}\right)C+\left(\alpha_w\tau_{c2}+\alpha_{c2}\right){C^2}}{\sqrt{1+\alpha_w^2\Sigma_w^2}}\right\}.
 \end{aligned}
\end{equation}
Comparing \eqref{eq: density-s-UC} and \eqref{eq: birdge-uc}, we observe that $\mathrm{pr}(S_t=1\mid U,C) $ and $E\{h(t, W,C)\mid U,C\}$ have the same parametric form. 
% 	$$\begin{array}{lcl}\dfrac{\psi_0+\psi_w\tau_0}{\sqrt{1+\psi_w^2\Sigma_w^2}}&=&\dfrac{\alpha_0+\alpha_w\tau_0}{\sqrt{1+\alpha_w^2\Sigma_w^2}},\;\begin{array}{lcl}\dfrac{\psi_w\tau_u+\psi_u}{\sqrt{1+\psi_w^2\Sigma_w^2}}&=&\dfrac{\alpha_w\tau_u}{\sqrt{1+\alpha_w^2\Sigma_w^2}}.\end{array}\\\dfrac{\psi_w\tau_{c1}+\psi_c}{\sqrt{1+\psi_w^2\Sigma_w^2}}&=&\dfrac{\alpha_w\tau_{c1}+\alpha_{c1}}{\sqrt{1+\alpha_w^2\Sigma_w^2}},\;\begin{array}{lcl}\dfrac{\psi_w\tau_{c2}}{\sqrt{1+\psi_w^2\Sigma_w^2}}&=&\dfrac{\alpha_w\tau_{c2}+\alpha_{c2}}{\sqrt{1+\alpha_w^2\Sigma_w^2}}.\end{array}\end{array}$$  
 %	{\red $$\alpha_w=\displaystyle \dfrac{\zeta + 0.25}{- 0.1875 \zeta^{2} - 0.09375 \zeta + 0.25}$$}
\end{itemize}  
%	$$\begin{array}{lcl}\operatorname{pr}(G=\operatorname{AT}\mid U,C)&=&E\{\operatorname{pr}(G=\operatorname{AT}\mid U,A,W,C)\mid U,C\}\\&=&\int\operatorname{pr}(G={\AT}\mid U,A,W,C)f(W\mid U,C)dW\\&=&\int\operatorname{pr}\left(q_{{\CO}}<G^{\ast}+\varezetalon\mid U,A,W,C\right)f(W\mid U,C)dW\\&=&\int\operatorname{pr}\left(\psi_wW+\varepsilon>t_{CO}-\psi_0-\psi_uU-\psi_cC\right)f(\varepsilon,W\mid U,C)dWd\varepsilon\\&=&\int\operatorname{pr}\left\{\psi_wW+\varepsilon-\psi_wE(W\mid U,C)\right.\\&&\;\;\;\;\;\;\;\left.<t_{CO}-\psi_0-\psi_uU-\psi_cC-\psi_wE(W\mid U,C)\right\}f(\varepsilon,W\mid U,C)dWd\varepsilon\\&=&\operatorname\Phi\left\{\dfrac{q_{{\CO}}-\psi_0-\psi_wE(W\mid U,C)-\psi_uU-\psi_cC}{\sqrt{1+\psi_w^2\Sigma_w^2}}\right\}\\&=&\operatorname\Phi\left[\dfrac{q_{{\CO}}-\psi_0-\psi_w\left\{\tau_0+\tau_uU+\tau_{c1}C+\tau_{c2}{C^2}\right\}-\psi_uU-\psi_cC}{\sqrt{1+\psi_w^2\Sigma_w^2}}\right]\\&=&\operatorname\Phi\left\{-\dfrac{(\psi_0+\psi_w\tau_0-q_{{\CO}})+\left(\psi_u+\psi_w\tau_u\right)U+\left(\psi_c+\psi_w\tau_{c1}\right)C+\psi_w\tau_{c2}{C^2}}{\sqrt{1+\psi_w^2\Sigma_w^2}}\right\}\end{array}$$  
 \subsection{Simulation details about other models}
 \begin{enumerate}
     \item We  now present the specific form of  $f(U\mid Z,A,W,C)$.
Since  $U \mid Z,A,W,C \sim N\{q_u( Z,A,W,C),\Sigma_u^2\}$,  where $\Sigma_u^2=\sigma_u^2(1-\rho^2_2)$ and the specific form of  $q_w(Z,A,C,U)$ is 
\begin{align*}
&~~~~~q_w(Z,A,W,C)=\nu_0+\nu_zZ+\nu_a A+\nu_{c1}C+\nu_{c2}C^2+\nu_wW,\\
  & \nu_0= \dfrac{\iota_0\sigma_w-\gamma_0\sigma_u\rho_2}{\sigma_w},~\nu_z= \frac{\iota_z \sigma_w - \gamma_z \sigma_u \rho_2}{\sigma_w},~\nu_a= \frac{\iota_a \sigma_w - \gamma_a \sigma_u \rho_2}{\sigma_w},\\&\nu_{c1}= \frac{\iota_{c1} \sigma_w - \gamma_{c1} \sigma_u \rho_2}{\sigma_w},\;\nu_{c2}= \frac{\iota_{c2} \sigma_w - \gamma_{c2} \sigma_u \rho_2}{\sigma_w}=0,\;\nu_u=\displaystyle \frac{\sigma_u \rho_2}{\sigma_w}.
\end{align*} 
\item The specific form of $\pr(S_t\mid Z,A,W,C)$:
  \begin{align*}  
&\mathrm{pr}(S_1=1\mid Z,A,W,C)\\&~~~~~~=E\left\{\mathrm{pr}(S_1=1\mid Z,U,A,W,C)\mid Z,A,W,C\right\}\\&~~~~~~=E\left\{\mathrm{pr}(S_1=1\mid U,W,C)\mid Z,A,W,C\right\}\\&~~~~~~=E\left\{\Phi(\zeta_0+\zeta_wW+\zeta_uU+\zeta_cC)\mid Z,A,W,C\right\}\\&~~~~~~=\operatorname\Phi\left\{\frac{\zeta_0+\zeta_wW+\zeta_uE(U\mid Z,W,A,C)+\zeta_cC}{\sqrt{1+\zeta_u^2\Sigma_u^2}}\right\}\\&~~~~~~=\operatorname\Phi\left\{\frac{\zeta_0+\zeta_wW+\zeta_u(\nu_0+\nu_zZ+\nu_aA+\nu_wW)+\zeta_cC}{\sqrt{1+\zeta_u^2\Sigma_u^2}}\right\}\\&~~~~~~=\operatorname\Phi\left\{\frac{\zeta_0+\zeta_u\nu_0+\zeta_u\nu_zZ+\zeta_u\nu_aA+(\zeta_w+\zeta_u\nu_w)W+(\zeta_u\nu_c+\zeta_c)C}{\sqrt{1+\zeta_u^2\Sigma_u^2}}\right\}.\end{align*}

\begin{align*}  &\operatorname{pr}(S_0=1\mid Z,A,W,C)\\&~~~~~~=E\{\operatorname{pr}(G=\mathrm{AT}\mid U,Z,A,W,C)\mid Z,A,W,C\}\\&~~~~~~=E\{\operatorname{pr}(G=\mathrm{AT}\mid U,W,C)\mid Z,A,W,C\}\\&~~~~~~=E\{\Phi(\zeta_0-\exp(\zeta_1)+\zeta_wW+\zeta_uU+\zeta_cC)\mid Z,A,W,C\}\\&~~~~~~=\operatorname\Phi\left\{\frac{\zeta_0-\exp(\zeta_1)+\zeta_wW+\zeta_uE(U\mid Z,W,A,C)+\zeta_cC}{\sqrt{1+\zeta_u^2\Sigma_u^2}}\right\}\\&~~~~~~=\operatorname\Phi\left\{\frac{\zeta_0+\zeta_u\nu_0-\exp(\zeta_1)+\zeta_u\nu_zZ+\zeta_u\nu_aA+(\zeta_w+\zeta_u\nu_w)W+(\zeta_u\nu_c+\zeta_c)C}{\sqrt{1+\zeta_u^2\Sigma_u^2}}\right\}. \end{align*}
%$$\begin{array}{lcl}\mathrm{pr}(S_t=1\mid Z ,A,W,C)=\operatorname\Phi\left\{\dfrac{\zeta_0+\zeta_u\nu_0 -\exp(\zeta_1)+\exp(\zeta_1)t +\zeta_u\nu_zZ+\zeta_u\nu_aA+(\zeta_w+\zeta_u\nu_w)W+(\zeta_u\nu_c+\zeta_c)C}{\sqrt{1+\zeta_u^2\Sigma_u^2}}\right\}.\end{array}$${\blue 
{For simplicity, we let
 $\mathrm{pr}(S_t=1\mid Z,A,W,C)=\operatorname\Phi\{\psi_0+\exp({\psi_1})(t-1)+\psi_zZ+\psi_aA+\psi_wW+\psi_cC\} $,  
 where $$\begin{array}{c}\psi_0=\dfrac{\zeta_0+\zeta_u\nu_0 }{\sqrt{1+\zeta_u^2\Sigma_u^2}},\;\psi_1=\log\left\{\dfrac{\exp(\zeta_1)}{\sqrt{1+\zeta_u^2\Sigma_u^2}}\right\},\;\psi_z=\dfrac{\zeta_u\nu_z}{\sqrt{1+\zeta_u^2\Sigma_u^2}},\\\psi_a=\dfrac{\zeta_u\nu_a}{\sqrt{1+\zeta_u^2\Sigma_u^2}},\;\psi_w=\dfrac{\zeta_w+\zeta_u\nu_w}{\sqrt{1+\zeta_u^2\Sigma_u^2}},\;\psi_c=\dfrac{\zeta_u\nu_c+\zeta_c}{\sqrt{1+\zeta_u^2\Sigma_u^2}}.\end{array}$$}
 \item  The specific form of $\pr(S_t\mid Z,A,C)$,
%$$\begin{array}{lcl}\mathrm{pr}(S_t=1\mid Z,A,C)&=&E\left\{\mathrm{pr}(S_t=1\mid Z,A,W,C)\mid Z,A,C\right\}\\&=&E\left[\operatorname\Phi\{\psi_0+\psi_1(1-t)+\psi_zZ+\psi_aA+\psi_wW+\psi_cC\}\mid Z,A,C\right]\\&=&\operatorname\Phi\left\{\dfrac{\psi_0+\psi_1(1-t)+\psi_zZ+\psi_aA+\psi_wE(W\mid Z,A,C)+\psi_cC}{\sqrt{1+\psi_w^2\sigma_w^2}}\right\}.\end{array}$$
\begin{align*}
\mathrm{pr}(S_t=1&\mid Z,A,C)\\& =E\left\{\mathrm{pr}(S_t=1\mid Z,A,W,C)\mid Z,A,C\right\}\\& =E\left[\operatorname\Phi\left\{\psi_0+\exp(\psi_1)(t-1)+\psi_zZ+\psi_aA+\psi_wW+\psi_cC\right\}\mid Z,A,C\right]\\& =\operatorname\Phi\left\{\dfrac{\psi_0+\exp(\psi_1)(t-1)+\psi_zZ+\psi_aA+\psi_wE(W\mid Z,A,C)+\psi_cC}{\sqrt{1+\psi_w^2\sigma_w^2}}\right\}.
\end{align*}%
 \end{enumerate} 
\newpage
 \subsection{Additional simulation results}
 \begin{table}[h]
    \caption{Simulation studies with bias ($\times 100$), standard error ($\times 100$)  and 95\% coverage probability ($\times 100$) for various settings and sample sizes.  }\label{tab: result-simulation-0}
\centering
\resizebox{0.95\columnwidth}{!}{% 
\begin{tabular}{cccrrrrrrrrrrrr}
\toprule[1pt]   $n$ &  Case & \multicolumn{1}{c}{$(\theta_a,\theta_w)$} &  & \multicolumn{3}{c}{$\Delta_\AT$} &  & \multicolumn{3}{c}{$\Delta_\CO$} &  & \multicolumn{3}{c}{$\Delta_\NT$} \\\addlinespace[0.5mm]\midrule[1pt]\addlinespace[1mm]    &&&&&\multicolumn{9}{c}{$\zeta_u=0$}   \\\addlinespace[0.5mm] \cline{5-15} \addlinespace[1mm]\addlinespace[0.5mm]         \addlinespace[0.5mm]    &&                  &  & \multicolumn{1}{c}{Bias}         &    \multicolumn{1}{c}{Sd}      & \multicolumn{1}{c}{CP}          &  & \multicolumn{1}{c}{Bias}         &    \multicolumn{1}{c}{Sd}      & \multicolumn{1}{c}{CP}      &  & \multicolumn{1}{c}{Bias}         &    \multicolumn{1}{c}{Sd}      & \multicolumn{1}{c}{CP}        \\\addlinespace[0.25mm]
$1000$     
            & (i) & $(0,0)    $                 &  & $-$0.9     & 12.4   & 94.8   &    & 1.5      & 42.7   & 96.2   &    & $-$1.0       & 35.3   & 96.0     \\
          
         & (ii)       & $(1,0)     $                &  & $-$0.5     & 13.7   & 95.8   &    & $-$3.1     & 51.4   & 97.8   &    & 2.1      & 29.3   & 97.6   \\
         & (iii)       & $(0,1)     $                &  & $-$7.2     & 28.8   & 97.0     &    & 18.9     & 67.4   & 96.0     &    & $-$4.1     & 28.0     & 94.6   \\
          & (iv)      & $(1,1)$                     &  & $-$3.4     & 33.1   & 95.8   &    & 5.9      & 71.8   & 97.4   &    & 0.3      & 31.3   & 95.6   \\ \addlinespace[0.5mm]      \addlinespace[1mm]    &         &         &  & \multicolumn{1}{c}{Bias}         &    \multicolumn{1}{c}{Sd}      & \multicolumn{1}{c}{CP}          &  & \multicolumn{1}{c}{Bias}         &    \multicolumn{1}{c}{Sd}      & \multicolumn{1}{c}{CP}      &  & \multicolumn{1}{c}{Bias}         &    \multicolumn{1}{c}{Sd}      & \multicolumn{1}{c}{CP}       \\\addlinespace[0.25mm]
$5000$      
           & (i)     & $(0,0)     $                &  & $-$0.2     & 5.5    & 95.6   &    & $-$0.2     & 21.0     & 96.0     &    & 0.6      & 14.9   & 95.0     \\
            
           & (ii)     & $(1,0)  $                   &  & $-$0.1     & 6.0      & 95.4   &    & $-$1.1     & 24.1   & 96.0     &    & 0.5      & 11.5   & 95.4   \\
           & (iii)     & $(0,1)   $                  &  & $-$3.2     & 13.5   & 96.6   &    & 7.1      & 31.9   & 97.2   &    & $-$0.2     & 12.3   & 94.6   \\
            & (iv)    & $(1,1)  $                   &  & $-$2.1     & 13.9   & 96.2   &    & 4.0        & 32.2   & 96.2   &    & 0.6      & 12.6   & 95.0     \\  \addlinespace[0.5mm]\midrule[1pt] \addlinespace[1mm]    &&&&&\multicolumn{9}{c}{$\zeta_u=0.1$}   \\\addlinespace[0.5mm] \cline{5-15}  \addlinespace[0.5mm]  \addlinespace[1mm]    &           &       &  & \multicolumn{1}{c}{Bias}         &    \multicolumn{1}{c}{Sd}      & \multicolumn{1}{c}{CP}          &  & \multicolumn{1}{c}{Bias}         &    \multicolumn{1}{c}{Sd}      & \multicolumn{1}{c}{CP}      &  & \multicolumn{1}{c}{Bias}         &    \multicolumn{1}{c}{Sd}      & \multicolumn{1}{c}{CP}       \\\addlinespace[0.25mm] $1000$    
            & (i)     & $(0,0) $                 &  & $-$0.8     & 9.7    & 95.0      &    & 2.1      & 40.9   & 96.4   &    & $-$2.0        & 30.2   & 95.8   \\ 
            & (ii)    & $(1,0) $                   &  & $-$0.6     & 10.2   & 95.2   &    & $-$0.6     & 47.4   & 96.8   &    & 0.0         & 24.2   & 96.8   \\
            & (iii)       & $(0,1) $                &  & $-$2.6     & 17.9   & 95.4   &    & 10.9     & 54.1   & 95.2   &    & $-$3.2     & 22.9   & 94.8   \\
            & (iv)       & $(1,1) $                &  & $-$1.1     & 18.7   & 96.6   &    & 3.0         & 56.5   & 97.6   &    & $-$0.4     & 26.5   & 96.2   \\ \addlinespace[0.5mm]  \addlinespace[1mm]    &                  &  & &\multicolumn{1}{c}{Bias}         &    \multicolumn{1}{c}{Sd}      & \multicolumn{1}{c}{CP}          &  & \multicolumn{1}{c}{Bias}         &    \multicolumn{1}{c}{Sd}      & \multicolumn{1}{c}{CP}      &  & \multicolumn{1}{c}{Bias}         &    \multicolumn{1}{c}{Sd}      & \multicolumn{1}{c}{CP}       \\\addlinespace[0.25mm]
$5000$      
            & (i)    & $(0,0) $                   &  & $-$0.3     & 4.2    & 94.0     &    & 0.2      & 19.4   & 95.8   &    & 0.9      & 13.7   & 94.0      \\
            
            & (ii)       & $(1,0) $                &  & $-$0.2     & 4.7    & 95.2   &    & $-$0.6     & 22.5   & 95.6   &    & 0.6      & 9.9    & 96.6   \\
            & (iii)      & $(0,1) $                 &  & $-$1.7     & 7.8    & 96.0      &    & 4.8      & 24.0      & 95.4   &    & 0.1      & 10.3   & 95.8   \\
            & (iv)       & $(1,1) $                &  & $-$1.2     & 8.0      & 95.8   &    & 3.0        & 24.6   & 95.6   &    & 0.6      & 10.6   & 96.6   \\     \addlinespace[0.5mm]\bottomrule[1pt] 
\end{tabular}
} 
\end{table}

 \begin{table}[h]
    \caption{Simulation studies with bias ($\times 100$), standard error ($\times 100$)  and 95\% coverage probability ($\times 100$) for various settings and sample sizes.  }\label{tab: result-simulation-0.4}
\centering
\resizebox{0.95\columnwidth}{!}{% 
\begin{tabular}{cccrrrrrrrrrrrr}
\toprule[1pt]   $n$ &  Case & \multicolumn{1}{c}{$(\theta_a,\theta_w)$} &  & \multicolumn{3}{c}{$\Delta_\AT$} &  & \multicolumn{3}{c}{$\Delta_\CO$} &  & \multicolumn{3}{c}{$\Delta_\NT$} \\\addlinespace[0.5mm]\midrule[1pt]\addlinespace[1mm]    &&&&&\multicolumn{9}{c}{$\zeta_u=0.3$} \\\addlinespace[0.5mm] \cline{5-15} \addlinespace[1mm]\addlinespace[0.5mm]         \addlinespace[0.5mm]  &                  &  && \multicolumn{1}{c}{Bias}         &    \multicolumn{1}{c}{Sd}      & \multicolumn{1}{c}{CP}          &  & \multicolumn{1}{c}{Bias}         &    \multicolumn{1}{c}{Sd}      & \multicolumn{1}{c}{CP}      &  & \multicolumn{1}{c}{Bias}         &    \multicolumn{1}{c}{Sd}      & \multicolumn{1}{c}{CP}     \\\addlinespace[0.25mm]
$1000$     
            & (i)    & $(0,0) $                   &  & $-$0.6     & 6.5    & 95.0     &    & 3.5      & 41.4   & 95.8   &    & $-$3.8     & 26.6   & 93.8   \\
           
            & (ii)       & $(1,0) $                &  & $-$0.5     & 7.0      & 96.0     &    & 1.7      & 47.5   & 96.0     &    & $-$1.4     & 20.2   & 96.0     \\
            & (iii)          & $(0,1) $             &  & $-$0.6     & 9.6    & 96.2   &    & 7.5      & 45.5   & 95.8   &    & $-$2.8     & 20.9   & 94.2   \\
            & (iv)        & $(1,1) $               &  & $-$0.3     & 9.6    & 95.8   &    & 1.9      & 48.8   & 97.0     &    & $-$1.0       & 22.1   & 96.0     \\ \addlinespace[0.5mm]        \addlinespace[1mm]    &                  &  && \multicolumn{1}{c}{Bias}         &    \multicolumn{1}{c}{Sd}      & \multicolumn{1}{c}{CP}          &  & \multicolumn{1}{c}{Bias}         &    \multicolumn{1}{c}{Sd}      & \multicolumn{1}{c}{CP}      &  & \multicolumn{1}{c}{Bias}         &    \multicolumn{1}{c}{Sd}      & \multicolumn{1}{c}{CP}       \\\addlinespace[0.25mm]
$5000$      
            & (i)        & $(0,0) $               &  & $-$0.3     & 3.2    & 94.8   &    & 1.7      & 19.3   & 95.8   &    & 0.8      & 12.4   & 95.4   \\
          
            & (ii)       & $(1,0) $                &  & $-$0.3     & 3.5    & 93.6   &    & 1.2      & 22.2   & 96.0     &    & 0.2      & 8.8    & 95.4   \\
            & (iii)            & $(0,1) $           &  & $-$0.7     & 4.3    & 94.4   &    & 4.0        & 20.1   & 93.4   &    & 0.1      & 9.1    & 95.8   \\
            & (iv)     & $(1,1) $                  &  & $-$0.6     & 4.3    & 94.0     &    & 3.2      & 21.3   & 94.4   &    & 0.3      & 9.5    & 96.4 \\    \addlinespace[0.5mm]\midrule[1pt] \addlinespace[1mm]   &&&&&\multicolumn{9}{c}{$\zeta_u=0.4$} \\\addlinespace[0.5mm] \cline{5-15} \addlinespace[1mm]\addlinespace[0.5mm]         \addlinespace[0.5mm]  &  &                &  & \multicolumn{1}{c}{Bias}         &    \multicolumn{1}{c}{Sd}      & \multicolumn{1}{c}{CP}          &  & \multicolumn{1}{c}{Bias}         &    \multicolumn{1}{c}{Sd}      & \multicolumn{1}{c}{CP}      &  & \multicolumn{1}{c}{Bias}         &    \multicolumn{1}{c}{Sd}      & \multicolumn{1}{c}{CP}     \\\addlinespace[0.25mm]$1000$     
            & (i)       & $(0,0) $                &  & $-$0.5     & 6.0      & 95.0     &    & 4.6      & 41.4   & 94.8   &    & $-$5.2     & 25.5   & 94.8   \\
         
            & (ii)        & $(1,0) $               &  & $-$0.4     & 6.5    & 95.4   &    & 4.2      & 47.2   & 95.0     &    & $-$2.2     & 19.4   & 96.0    \\
            & (iii)        & $(0,1) $               &  & $-$0.5     & 8.0      & 97.0     &    & 8.2      & 44.7   & 95.0     &    & $-$3.2     & 20.2   & 94.2   \\
            & (iv)     & $(1,1) $                  &  & $-$0.4     & 8.0      & 96.6   &    & 2.9      & 47.9   & 96.8   &    & $-$1.4     & 21.8   & 96.4   \\ \addlinespace[0.5mm]        \addlinespace[1mm]    &          &        &  & \multicolumn{1}{c}{Bias}         &    \multicolumn{1}{c}{Sd}      & \multicolumn{1}{c}{CP}          &  & \multicolumn{1}{c}{Bias}         &    \multicolumn{1}{c}{Sd}      & \multicolumn{1}{c}{CP}      &  & \multicolumn{1}{c}{Bias}         &    \multicolumn{1}{c}{Sd}      & \multicolumn{1}{c}{CP}       \\\addlinespace[0.25mm]
$5000$     
            & (i)         & $(0,0) $              &  & $-$0.3     & 2.8    & 94.6   &    & 3.0        & 21.0     & 94.4   &    & 0.3      & 12.3   & 94.0     \\
           
            & (ii)      & $(1,0) $                 &  & $-$0.3     & 3.1    & 94.2   &    & 2.7      & 24.2   & 95.0     &    & 0.0        & 8.7    & 95.0     \\
            & (iii)      & $(0,1) $                 &  & $-$0.7     & 3.6    & 94.4   &    & 5.2      & 21.2   & 92.6   &    & $-$0.1     & 9.0      & 95.0     \\
            & (iv)        & $(1,1) $               &  & $-$0.6     & 3.7    & 94.6   &    & 4.0        & 22.4   & 93.6   &    & 0.3      & 9.5    & 95.4      \\  \addlinespace[0.5mm]\bottomrule[1pt] 
\end{tabular}
} 
\end{table} 					
				\end{document}